\documentclass{article}

\usepackage[margin=0.7in]{geometry}
\usepackage[english]{babel}
\usepackage[utf8]{inputenc}
\usepackage{graphicx}
\usepackage{graphics}
\usepackage{fancyhdr}
\usepackage{amsmath} \numberwithin{equation}{section}
\usepackage{amsthm}
\usepackage{mathrsfs}
\usepackage{mathtools}
\usepackage{slashed}
\usepackage{amssymb}
\usepackage{braket}
\usepackage{bm}
\usepackage{wrapfig}
\usepackage{adjustbox}
\usepackage{multicol}
\usepackage{float}
\usepackage{url}
\usepackage{xcolor}
\usepackage{pagecolor}
\usepackage{subcaption}
\usepackage{tensor}
\usepackage{multirow}
\usepackage{xcolor,colortbl}
\usepackage{comment}
\usepackage[colorlinks=true, allcolors=blue]{hyperref}

\usepackage{tikz}
\usetikzlibrary{arrows,matrix,positioning}
\usepackage{pst-node}
\usepackage{easybmat}
\usepackage{dsfont}
\newcommand{\C}{{\mathbb C}}
\newcommand{\R}{{\mathbb R}}
\newcommand{\Z}{{\mathbb Z}}
\newcommand{\Hq}{{\mathbb H}}

\newcommand{\id}{{\mathbb I}}
\newcommand{\im}{{\rm i\,}}

\newcommand{\be}{\begin{eqnarray}}
\newcommand{\ee}{\end{eqnarray}}

\title{Notes on Spinors and Polyforms I: General Case}
\author{Niren Bhoja and Kirill Krasnov\\ {}\\
{\it School of Mathematical Sciences, University of Nottingham, NG7 2RD, UK}}

\theoremstyle{definition}
\newtheorem{definition}{Definition}[section]
\newtheorem{proposition}{Proposition}[section]
\newtheorem{theorem}{Theorem}[section]
\newtheorem{lemma}{Lemma}[section]

\begin{document}
\maketitle

\begin{abstract}It is well-known that the Clifford algebra ${\rm Cl}(2n)$ can be given a description in terms of creation/annihilation operators acting in the space of inhomogeneous differential forms on $\C^n$. We found it convenient to refer to such inhomogeneous differential forms as polyforms. Geometrically, the construction proceeds by choosing a complex structure $J$ on $\R^{2n}$. Spinors are then polyforms on one of the two totally-isotropic subspaces $\C^n$ that arise as eigenspaces of $J$. There is a similar and also well-known description in the split signature case ${\rm Cl}(n,n)$, with differential forms now being those on $\R^n$. In this case the model is constructed by choosing a paracomplex structure $I$ on $\R^{n,n}$, and spinors are polyforms on one of the totally null eigenspaces $\R^n$ of $I$. The main purpose of the paper is to describe the geometry of an analogous construction in the case of a general Clifford algebra ${\rm Cl}(r,s), r+s=2m$. We show that in general a creation/annihilation operator model is in correspondence with a new type of geometric structure on $\R^{r,s}$, which provides a splitting $\R^{r,s}=\R^{2k,2l}\oplus \R^{n,n}$ and endows the first factor with a complex structure and the second factor with a paracomplex structure. We refer to such geometric structure as a mixed structure. It can be described as a complex linear combination $K=I+\im J$ of a paracomplex and a complex structure such that $K^2=\id$ and $K\bar{K}$ is a product structure. In turn, the mixed structure is in correspondence with a pair of pure spinors whose null subspaces are the eigenspaces of $K$. The conclusion is then that there is in general not one, but several possible creation/annihilation operator models for a given Clifford algebra. The number of models is the number of different types of pure spinors (distinguished by the real index, see the main text) that exists in a given signature. To illustrate this geometry, we explicitly describe all the arising models for ${\rm Cl}(r,s)$ with $r\geq s, r+s=2m\leq 6$.
\end{abstract}

\section{Introduction}

Given a (real) vector space $V$ with a metric $g(\cdot,\cdot)$ of signature $(r,s)$\footnote{Our convention is that $(r,s)$ means $r$ pluses and $s$ minuses.}, the Clifford algebra ${\rm Cl}_V\equiv {\rm Cl}(r,s)$ is the (non-commutative but associative) algebra generated by $V$ subject to the relation $v\circ v = g(v,v) \id$. We use $\circ$ to denote the Clifford product. The group ${\rm Pin}(r,s)$ is the subset of the Clifford algebra ${\rm Cl}(r,s)$ generated by vectors $v: g(v,v)=\pm 1$ of norm squared plus minus one. The group ${\rm Spin}(r,s)$ is subgroup of ${\rm Pin}(r,s)$ generated by an even number of such vectors. The Clifford algebra ${\rm Cl}(r,s)$ admits a representation in the (in general complex) space $S$ of spinors. We will assume $r+s=2m$ is even. The case of odd $r+s$ is easy to describe once the even case $r+s-1$ is understood. The (in general complex) dimension of the space of spinors is ${\rm dim}(S)=2^m$. 

Clifford algebras are isomorphic to one or two copies of matrix algebras $M_N(\R), M_N(\C)$ and $M_N(\Hq)$. There are two different ways to think about Clifford algebras and to obtain their classification. The most familiar and widely used in the literature way uses the tensor product construction, whose origin can be traced to Brauer and Weyl \cite{BrauerWeyl}. The second, more geometric approach based on maximal totally isotropic subspaces is due to Cartan \cite{Cartan}. The two approaches were combined by Chevalley \cite{Chevalley}. Very interesting historical remarks on spinors are available in \cite{BT} and \cite{TT}.

Let us start by briefly reminding the tensor product construction. Concretely, this can be given the form of Lemma 11.17 from \cite{Harvey}. We have\footnote{Our signature convention is opposite to that in \cite{Harvey}, which explains the flipping of $r,s$ as compared to this reference.}
\be\label{tensor-product}
{\rm Cl}(r+1,s+1) \cong {\rm Cl}(r,s) \otimes_\R M_2(\R), \\ \nonumber
{\rm Cl}(s+2,r) \cong {\rm Cl}(r,s) \otimes_\R M_2(\R), \\ \nonumber
{\rm Cl}(r,s+2) \cong {\rm Cl}(r,s) \otimes_\R \Hq.
\ee
These facts can be proved by an explicit construction. Thus, let $\gamma^I$ be the generators of ${\rm Cl}(r,s)$. Then the matrices
\be
\Gamma^I := \left( \begin{array}{cc} 0 & \gamma \\ \gamma & 0 \end{array}\right), \quad \Gamma^{I+1}:=\left( \begin{array}{cc} \id &0 \\ 0 & -\id \end{array}\right), \quad 
\Gamma^{I+2}:=\left( \begin{array}{cc} 0 & \id \\ -\id & 0 \end{array}\right)
\ee
anti-commute, and $(\Gamma^{I+1})^2=\id, (\Gamma^{I+2})^2=-\id$. Thus, these matrices generate ${\rm Cl}(r+1,s+1)$. 

A suitable modification of this construction gives the second line in (\ref{tensor-product}). Indeed, we can instead define
\be
\Gamma^I := \left( \begin{array}{cc} 0 & \gamma \\ -\gamma & 0 \end{array}\right), \quad \Gamma^{I+1}:=\left( \begin{array}{cc} \id & 0  \\ 0 & -\id   \end{array}\right), \quad 
\Gamma^{I+2}:=\left( \begin{array}{cc} 0 & \id \\ \id & 0 \end{array}\right).
\ee
Then $\Gamma^I$ generate the Clifford algebra ${\rm Cl}(s,r)$ of opposite signature, and both $\Gamma^{I+1}, \Gamma^{I+2}$ square to plus the identity.

The last line in (\ref{tensor-product}) is proved by the following construction
\be
\Gamma^I := \left( \begin{array}{cc} 0 & \gamma \\ \gamma & 0 \end{array}\right), \quad \Gamma^{I+1}:=\left( \begin{array}{cc} {\bf i} & 0    \\ 0 & -{\bf i}     \end{array}\right), \quad 
\Gamma^{I+2}:=\left( \begin{array}{cc} {\bf j} & 0 \\ 0 & -{\bf j}  \end{array}\right).
\ee
Here ${\bf i, j}$ are two imaginary quaternions. These matrices anti-commute and $\Gamma^{I+1}, \Gamma^{I+2}$ square to minus the identity.

One can then generate all Clifford algebras using (\ref{tensor-product}) and the easily verifiable facts
\be\label{seed}
{\rm Cl}(0,1)\cong \C, \quad {\rm Cl}(1,0) \cong \R\oplus \R, \\ \nonumber
{\rm Cl}(0,2)\cong \Hq, \quad {\rm Cl}(1,1) \cong M_2(\R), \quad {\rm Cl}(2,0) \cong M_2(\R).
\ee

To carry out the above construction in the case of concrete ${\rm Cl}(r,s)$ one needs to determine the seed from (\ref{seed}), as well as the sequence of tensor product transformations that realise ${\rm Cl}(r,s)$ according to (\ref{tensor-product}). This is a case by case construction.

This paper is devoted to the second approach to Clifford algebras and spin groups, which is due to Elie Cartan \cite{Cartan}, Brauer-Weyl \cite{BrauerWeyl} and Chevalley \cite{Chevalley}. It is based on maximally-isotropic subspaces and creation/annihilation operators. A more modern treatment of this approach is available in e.g. \cite{Harvey}. While this approach to Clifford algebras is standard in the mathematics literature, this is much less the case in the physics literature.  Indeed, the literature on e.g. grand unification, see e.g. \cite{DiLuzio:2011mda}, continues to use the tensor product construction of the Clifford algebra, which becomes very cumbersome for large spin groups. On of the main goals of this paper is to popularise Cartan's approach by providing details that may not be readily available. For large spin groups further economy is achieved by using the octonions. This is described in the accompanying paper \cite{BK2}.

The description of Clifford algebras we develop is based on creation-annihilation operators. Thus, it is known that both the case of ${\rm Cl}(2n,0)$ and ${\rm Cl}(n,n)$ can be very efficiently described with the technology of the creation-annihilation operators. In the case of ${\rm Cl}(2n,0)$ the construction proceeds by choosing a complex structure $J: J^2=-\id$ on $\R^{2n}$, thus identifying $\R^{2n}\sim \C^n$. The complex vector space $\C^n$ here should be thought of as the maximally-isotropic subspace of $\R^{2n}_\C$ that arises as one of the two eigenspaces of $J$. The $\gamma$-matrices can then be described as appropriate (complex) linear combinations of the creation and annihilation operators that act in the space of differential forms $\Lambda(\C^n)$. We remind this construction in details below. 

In the case of ${\rm Cl}(n,n)$ one similarly chooses a maximally isotropic subspace of $\R^{n,n}$, which in this case can be taken to be a copy of $\R^n$ (but other choices are also possible, see below). One then constructs $\gamma$-matrices as appropriate sums and differences of the creation-annihilation operators, now acting in in the space of real differential forms $\Lambda(\R^n)$. 

For the case of a general ${\rm Cl}(r,s)$, it is clear that one can realise it as a subalgebra of ${\rm Cl}_\C(2p) = {\rm Cl}(p,p)\otimes_\R \C$. Thus, one can complexify the creation/annihilation model of ${\rm Cl}(n,n)$ and then take an appropriate real slice that would correspond to ${\rm Cl}(r,s)$. While it is clear that this is possible, there are subtleties that arise and we do not know of any treatment that would do them justice. Some aspects of the real case are treated in Section 12 of \cite{Harvey}. However, there is some beautiful geometry that arises in treating the real case, and this does not appear to be readily available. Our other main aim in the present paper is to develop the geometry of the real case ${\rm Cl}(r,s)$ in sufficient details. 

The main novelty in our treatment is as follows. The creation/annihilation model of ${\rm Cl}(2n)$ arises after a complex structure on $\R^{2n}$ is chosen. Similarly, the "real" model of the split signature Clifford algebra ${\rm Cl}(n,n)$ arises after a paracomplex structure (see the main text) on $\R^{n,n}$ is chosen. We describe the generalisation of this statement to the case of ${\rm Cl}(r,s)$. Thus, we show that a creation/annihilation model of a general ${\rm Cl}(r,s)$ arises after a certain new type structure that mixes complex and paracomplex structures is chosen. We propose the following
\begin{definition} A {\bf mixed structure} on a real vector space $V$ equipped with a metric $g$ is a  linear map $K: V_\C\to V_\C$ satisfying $K^2=\id, K\bar{K}=\bar{K}K$ and $g(KX,KY)=g(\bar{K}X,\bar{K}Y)=-g(X,Y)$. Here $V_\C$ is the complexification $V_\C=\C\otimes_\R V$ and $\bar{K}$ is the complex conjugate of $K$.
\end{definition} 
We then have the following
\begin{proposition} Let $K: V_\C\to V_\C$ be a mixed structure on a real vector space $V$. Then the operator $P:=K\bar{K}=\bar{K} K$ is real $P: V\to V$, satisfies $P^2=\id$, and endows $V$ with (an orthogonal) product structure $V=V^+\oplus V^-$. Here $V^\pm$ are the eigenspaces of eigenvalue $\pm 1$ of $P$. The operator $K$ acts on $V^+$ as a paracomplex structure, and on $V^-$ as the imaginary unit times a complex structure. 
\end{proposition}
Thus, a mixed structure on $\R^{r,s}$ provides its decomposition into a direct sum of $\R^{2k,2l}\oplus\R^{m,m}$ and selects a complex structure on $\R^{2k,2l}$, as well as a paracomplex structure on $\R^{m,m}$. A mixed structure also selects two maximal totally null subspaces of $\R^{r,s}$, with $m$ real and $k+l$ complex null vectors. These arise as eigenspaces of $K$. A creation/annihilation model of ${\rm Cl}(r,s)$ only arises after a mixed structure is chosen. Each model also comes with two preferred pure spinors that have the property that their annihilator subspaces in  $\R^{r,s}$ are precisely the eigenspaces of $K$. 

The correspondence between a creation/annihilation operator model and a pair of pure spinors also works in the opposite direction. Thus, taking a pair of pure spinors $\psi_{1,2}$ that have a non-vanishing inner product $\langle\psi_1,\psi_2\rangle\not=0$, the annihilator subspaces of $\psi_{1,2}$ in $\R^{2k,2l}$ are complementary to each other. There then exists a structure $K$ of the mixed type that has the annihilator subspaces of $\psi_{1,2}$ as its eigenspaces. It can be obtained explicitly by computing the 2-form $\langle \psi_1, \Gamma\Gamma\psi_2\rangle$. Here two copies of $\Gamma$-matrices are inserted between the pure spinors $\psi_{1,2}$. Raising one of the indices of the arising 2-form with the metric on $\R^{r,s}$ one obtains an operator that is a multiple of $K$ with the desired properties. The availability of this construction nicely correlates with the fact, described in \cite{KT}, that there are in general different types of maximally-isotropic subspaces of $\R^{r,s}$, distinguished by their real index, and thus different types of pure spinors. Choosing a pair of pure spinors $\psi_{1,2}$  of real index $m$ satisfying $\langle\psi_1,\psi_2\rangle\not=0$ one generates a structure $K$ of the mixed type that has the property that its eigenspaces contain $m$ real basis null vectors. 

Thus, one of the main points of this paper is that a general ${\rm Cl}(r,s)$ can be described by a number of different models. Each model corresponds to a different type of pure spinors that exist for ${\rm Spin}(r,s)$. One of the models for ${\rm Cl}(r,s)$ that we describe corresponds to what in \cite{KT} is called simpler simple spinors. They are the pure spinors whose real index (see below for the definition) is maximal. In this case of the maximal real index, the construction we describe is essentially known, see Section 12 of \cite{Harvey}, in particular the section titled "The pinor reality map". We are not aware of the description of the other possible creation/annihilation operator models. 

In physics applications it is usually sufficient to have a description of the corresponding spin groups ${\rm Spin}(r,s)$, rather than the Clifford algebra ${\rm Cl}(r,s)$. While Clifford algebras ${\rm Cl}(r,s)$ and ${\rm Cl}(s,r)$ are distinct, there is no such difference at the level of the spin groups. Thus, if one only cares about ${\rm Spin}(r,s)$, one can assume $r\geq s$. This simplifies some of the constructions, as there are less cases to consider. We will always assume $r\ge s$ in the present paper. The Clifford algebras ${\rm Cl}(r,s)$ with $r<s$ can be straightforwardly considered by our methods as well, but we will leave this out to simplify considerations.

The organisation of this paper is as follows. We start by reviewing, in Section \ref{sec:2n}, the creation/annihilation operator construction of the compact case ${\rm Spin}(2n)$. This is the most well-known case. We also describe here the known geometrical relation between the pure spinors and complex structures. We give an analogous treatment of the split case ${\rm Spin}(n,n)$ in Section \ref{sec:n-n}. 
We review the necessary for us constructions from \cite{KT} in Section \ref{sec:MTN}. In particular, the notion of the real index of a maximally-isotropic subspace of $\R^{r,s}$ is described here. 
We then proceed to describe the general case of the creation/annihilation operator construction in Section \ref{sec:gen}. Each such construction is in a correspondence with a choice of two complementary maximally-isotropic subspaces of $\R^{r,s}$. Alternatively, each model is in a correspondence with a new type of geometric structure that can be put on $\R^{r,s}$, which mixes complex and paracomplex structures, and which we describe here. We also describe the reality map applicable to each case, and thus give a classification of Majorana and Majorana-Weyl spinors. The purpose of Sections \ref{sec:two}, \ref{sec:four}, \ref{sec:six} is to explicitly carry out the creation/annihilation operator constructions of ${\rm Spin}(r,s)$ for $r+s\leq 6$. We conclude with a discussion. 

There exists a link between the creation/annihilation operator construction of this paper and the description of Clifford algebras using quaternions and octonions. This is developed in the accompanying paper.

\section{Polyform Representations of Spin($2n$)}
\label{sec:2n}

\subsection{Clifford algebra}

Let us start with the Clifford Algebra Cliff$_{2n}$. To construct it, we introduce a complex structure $J:\R^{2n}\to \R^{2n}, J^2=-\id$ on $\R^{2n}$ and identify it with $\C^n\oplus \overline{\C^n}$. Concretely, we think of $\C^n$ as the eigenspace of $J$ of eigenvalue $-\im$. 

We then consider the space $\Lambda \mathbb{C}^n$ of mixed degree (that is inhomogeneous) differential forms on $\C^n$. We found it convenient to refer to such mixed degree forms as {\bf polyforms}. Let $e_i, i=1,\ldots, n$ be a basis in $\Lambda^1 \C^n$. We introduce the operators $a_i$ of creation and $a^\dagger_i$ of annihilation of basic 1-forms $e_i$. That is 
\be\label{cr-an}
a_i \omega := e_i \wedge \omega, \qquad a^\dagger_i \omega := e_i \lrcorner\, \omega,
\ee
where $\omega\in \Lambda\C^n$, $\wedge$ is the usual wedge product, and $ e_i \lrcorner$ is the operator that looks for an $e_i$ factor in $\omega$ and deletes it:
\be
e_i \lrcorner\, \left(e_{i_1} \wedge \ldots \wedge e_{i_k}\right)  =\sum_{m=1}^k  (-1)^{m-1} \delta_{i i_m} e_{i_1} \wedge \ldots ({\mathrm{omit\,\, mth\,\, factor}}) \ldots \wedge e_{i_k}.
\ee
The introduced creation/annihilation operators satisfy the
following anti-commutator relations
\be
\{ a_i, a_j^\dagger \} = \delta_{ij}.
\ee

We construct the Clifford generators, or Gamma matrices, as the appropriate linear combinations of these creation/annihilation operators:
\begin{equation}
\Gamma_i := 
a_i + a_i^\dagger ,\qquad 
    \Gamma_{i+n} := i (a_i - a_i^\dagger), \qquad 
    1 \leq i \leq n.
\end{equation}
It is then easy to verify that the Gamma matrices satisfy the Clifford algebra relations
\begin{equation}\label{clifford relation 2n}
    \Gamma_A\Gamma_B+\Gamma_B\Gamma_A=2\delta_{AB}, \qquad 1 \leq A,B \leq 2n.
\end{equation}
We note that there is some ambiguity in the above construction, in that we could have instead introduced the factors of the imaginary unit in the $a_i+a_i^\dagger$ operators, rather than in $a_i-a_i^\dagger$. Then all our $\Gamma$-matrices would square to minus the identity instead. This difference is important in the world of Clifford algebras, but is immaterial at the level of the spin algebra that we discuss next. Nevertheless, in the general ${\rm Cl}(r,s)$ case this ambiguity becomes important and is related to the different possible types of models that can be constructed. This will be described in due course. 

\subsection{Spin Lie algebra, semi-spinors}

The (spinor representation of the) Lie algebra  $\mathfrak{spin}(2n)$ is generated by the commutators of $\Gamma$-matrices. Alternatively, given an anti-symmetric $2n\times 2n$ matrix $X^{AB}$ we can form the following operator acting on polyforms
\be\label{spin-2n}
A(X):= \frac{1}{4} \sum_{A<B} X^{AB} \Gamma_A \Gamma_B.
\ee
The map $A(X)$ introduced is the Lie algebra homomorphism in the sense that 
\be
[A(X), A(Y)]=A([X,Y]),
\ee
where on the right-hand side $[X,Y]$ is the commutator of two anti-symmetric matrices $X,Y$, i.e. $[X,Y]^{AB}=X^A{}_C Y^{CB}- Y^{A}{}_C X^{CB}$, and the index is lowered with the metric $\delta_{AB}$ on $\R^{2n}$.

The operators $A(X)$ act on the space $\Lambda(\C^n)$ of polyforms, and preserve the subspaces of even and odd degree polyforms. Thus, the space of polyforms split
\begin{equation}
    \Lambda \mathbb{C}^{n}=\Lambda^{\text{even}}\mathbb{C}^n\oplus\Lambda^{\text{odd}}\mathbb{C}^n,
\end{equation}
and each subspace is a representation of $\mathfrak{spin}(2n)$. We shall refer to the space of all polyforms as the space of {\bf spinors}. We will use the notation $S$ for this space, so our construction of the Clifford algebra identifies
\be
S=\Lambda \C^n.
\ee
The even and odd degree polyforms will be referred to as {\bf semi-spinors}. We will call the even (odd) polyforms {\bf positive (negative)} semi-spinors. We will use the notation
\be
S^\pm = \Lambda^{even/odd} \C^n,
\ee
and $S=S_+ \oplus S_-$.

Given the Clifford algebra, one can describe not just the spin Lie algebra, but also the group ${\rm Spin}(2n)$. However, for most physics applications the Lie algebra is sufficient. 

\subsection{Inner product}

The polyform description of spinors we are developing allows for a very simple description of the $\mathfrak{spin}(2n)$-invariant inner product on $S$. Let us introduce the operation $\sigma$ that rewrites each decomposable polyform in the opposite order
\be
\sigma( e_{i_1}\wedge e_{i_2} \wedge \ldots \wedge e_{i_k}) = e_{i_k} \wedge \ldots \wedge e_{i_2} \wedge e_{i_1}.
\ee
This operator extends to all of $\Lambda \C^n$ by linearity. 

Then, given two polyforms $\psi_1, \psi_2\in S$, the inner product is defined as
\be\label{inner-prod}
\langle \psi_1, \psi_2 \rangle = \sigma(\psi_1) \wedge \psi_2 \Big|_{top},
\ee
where the meaning of the right-hand side is that the wedge product of two polyforms is taken and then the projection to the top degree is applied. Thus, the inner product is defined only after a top degree element from $\Lambda^n \C^n$ is chosen. To put it differently, an invariant inner product is only defined up to multiplication by a (complex-valued) constant. 

To prove the invariance of (\ref{inner-prod}) under $\mathfrak{spin}(2n)$ transformations we first need to establish the adjointness properties of the creation/annihilation operators. We have
\be
\langle a_i \psi_1, \psi_2 \rangle=\langle \psi_1, a_i \psi_2 \rangle, \qquad \langle a_i^\dagger \psi_1, \psi_2 \rangle=\langle \psi_1, a_i^\dagger \psi_2 \rangle.
\ee
So, both the creation and annihilation operators are self-adjoint with respect to the inner product (\ref{inner-prod}). This means that also the $\Gamma$-matrices are self-adjoint. But then the adjoint of a product of two $\Gamma$-matrices is their product written in the opposite order, which for distinct $\Gamma$-matrices is minus the original product. This shows that
\be 
\langle A(X) \psi_1, \psi_2\rangle + \langle \psi_1, A(X) \psi_2 \rangle =0,
\ee
and the inner product is invariant. 

\subsection{Reality conditions, Majorana spinors}
\label{sec:RR'}

We now introduce two anti-linear maps $R, R'$ on $S$, which either commute or anti-commute with all $\Gamma$-matrices. As the result both $R, R'$ commute with all Lie algebra operators. Depending on $n$, these operators square to either plus or minus the identity operator. When we have an anti-linear operator that squares to the identity and commutes with all Lie algebra transformations, it is meaningful to restrict the action of $\mathfrak{spin}(2n)$ to one of the two eigenspaces of the anti-linear operator. This is how Majorana spinors arise.

We thus define the following anti-linear maps, given either by the product of all "real" $\Gamma$-matrices followed by the complex conjugation, or by the product of all "imaginary" $\Gamma$-matrices again followed by the complex conjugation
\begin{equation}
    R \coloneqq \Gamma_{1} \ldots \Gamma_{n} \ast, \qquad R' \coloneqq \Gamma_{n+1} \ldots \Gamma_{2n} \ast .
\end{equation}
Here $\ast$ is the complex conjugation map, that is $\ast z \coloneqq z^{*}$ for any $z \in \mathbb{C}^n$. 

We have the following lemma
\begin{lemma} 
The operators $R,R'$ either commute or anto-commute with all the $\Gamma$-operators
   \begin{equation}
    R \Gamma_A = (-1)^{n-1}\Gamma_A R, \qquad R' \Gamma_A = (-1)^{n}\Gamma_A R' \qquad A \in \{1, \ldots, 2n\}
    \end{equation} 
\end{lemma}
This means that both $R,R'$ are anti-linear operators that commute with all operators from $\mathfrak{spin}(2n)$. The proof is by verification. 

We now need to establish a result on the square of each map
\begin{lemma}
    \begin{equation}\label{R2}
    R^2 = (-1)^{\frac{n(n-1)}{2}}\cdot\textbf{1}  \qquad
    (R')^2 = (-1)^{\frac{n(n+1)}{2}}\cdot\textbf{1} 
    \end{equation}
\end{lemma}

\begin{proof}
We have
\begin{center}
\begin{minipage}{.5\textwidth}
\begin{equation}
    \begin{split}
        R^2 &= \Gamma_{1} \ldots \Gamma_{n} \ast \Gamma_{1} \ldots \Gamma_{n} \ast = \\ &= \Gamma_{1} \ldots \Gamma_{n} \Gamma_{1} \ldots \Gamma_{n} \overbrace{\ast \ast}^{=1} = \\ &= (-1)^{n-1} \Gamma_1^2 \ldots \Gamma_n \Gamma_2 \ldots \Gamma_n = \\ &= (-1)^{(n-1) + (n-2) + \cdots + 1} \Gamma_1^2 \dots \Gamma_n^2 = \\ &= (-1)^{\frac{n(n-1)}{2}} \cdot \textbf{1}
    \end{split}    
\end{equation}
\end{minipage}%
\begin{minipage}{.5\textwidth}
\begin{equation}
    \begin{split}
        (R')^2 &= \Gamma_{n+1} \ldots \Gamma_{2n} \ast \Gamma_{n+1} \ldots \Gamma_{2n} \ast  \\ &= (-1)^{n} \Gamma_{n+1} \ldots \Gamma_{2n} \Gamma_{n+1} \ldots \Gamma_{2n} \overbrace{\ast \ast}^{=1} = \\ &= (-1)^{n+(n-1)} \Gamma_{n+1}^2 \ldots \Gamma_{2n} \Gamma_{n+2} \ldots \Gamma_{2n} = \\ &= (-1)^{n + (n-1) + \cdots + 1} \Gamma_{n+1}^2 \dots \Gamma_{2n}^2 = \\ &= (-1)^{\frac{n(n+1)}{2}} \cdot \textbf{1}
    \end{split}
\end{equation}
\end{minipage}
\end{center}
\end{proof}
We thus see that when $n$ is even $R^2=(R')^2$, and when $n$ is odd $R^2=-(R')^2$. Thus, when $n\in 4\mathbb{Z}$ both $R,R'$ square to plus the identity. When $n$ is even but not a multiple of 4, both $R,R'$ square to minus the identity, and neither gives a reality condition. When $n$ is odd either $R$ or $R'$ gives a reality condition. 

We have the following lemma
\begin{lemma}
$R$ and $R'$, up to a phase, are the only possible reality conditions.
\end{lemma}
The statement here is that, up to multiplication by a complex number, $R,R'$ are the only anti-linear operators that either commute or anti-commute with all Cliff$_{2n}$. A proof is analogous to the proof of Lemma 12.75 in \cite{Harvey}. The phase can be chosen so as to have a convenient reality property for the spinor inner product, see Lemma 12.90 in \cite{Harvey}, but we will not use this in our treatment, always working with either $R$ or $R'$ in this paper.

When there is a reality condition on $S$, i.e. an anti-linear operator $\mathcal R: {\mathcal R}^2=\id$ that commutes with all operators from $\mathfrak{spin}(2n)$, we can restrict the action of $\mathfrak{spin}(2n)$ to the {\bf Majorana} spinors, which is the subspace $S_M=\{ \psi \in S: {\mathcal R} \psi = \psi\}$. We have seen that there are no Majorana spinors when $n$ is even but not a multiple of four. 

We can also go a step further, and ask whether there exists a choice of the reality condition that can be imposed on the spaces of semi-spinors $S^\pm$. If such real semi-spinors exist they are called {\bf Majorana-Weyl spinors}. It is clear that only when $n$ is even both $R,R'$ are given by a product of an even number of $\Gamma$-matrices, and thus preserve $S^\pm$. However, we have seen that only when $n\in 4\mathbb{Z}$ we have a reality condition. 

The above discussion can be summarised as follows. When $n$ is odd, we have Majorana spinors. When $n\in 4\mathbb{Z}$ we have Majorana-Weyl spinors. When $n$ is even $n\not\in 4\mathbb{Z}$ there are no Majorana spinors. 

\subsection{Distinguished ${\mathfrak u}(n)$ subalgebra}

The described creation/annihilation operator model for $\mathfrak{spin}(2n)$ starts by introducing an orthogonal complex structure $J$ on $\R^{2n}$, i.e. a complex structure that is compatible with the usual Euclidean metric on $\R^{2n}$, i.e. $\delta(J\cdot, J\cdot)=\delta(\cdot,\cdot)$. The eigenspaces of $J$ are totally null, and we have constructed spinors as polyforms generated by the vectors in the eigenspace of $J$ of eigenvalue $-\im$, i.e. $\C^n : J \C^n = -\im \C^n$. 

It is then clear that our creation/annihilation operator model for $\mathfrak{spin}(2n)$ comes with a distinguished ${\mathfrak u}(n)$ subalgebra. This is the subalgebra of $\mathfrak{spin}(2n)$ that does not mix the polyforms of different degrees in $\Lambda(\C^n)$. Its alternative characterisation is that it is generated by operators $a_i a_j^\dagger$ containing one creation and one annihilation operator. To see how this arises explicitly, let us rewrite the general Lie algebra element using indices of dimension $n$ rather than $2n$. We have
\be
A_{\mathfrak{spin}(2n)} = \frac{1}{2} X^{ij} (a_i + a_i^\dagger) (a_j + a_j^\dagger) - \frac{1}{2} \tilde{X}^{ij}  (a_i - a_i^\dagger) (a_j -a_j^\dagger) +\im  Y^{ij} (a_i + a_i^\dagger) (a_j - a_j^\dagger).
\ee
The matrices $X^{ij}, \tilde{X}^{ij}$ are anti-symmetric, while $Y^{ij}$ does not have any symmetry. The summation convention is implied. The subalgebra of this that does not contain the products $a_i a_j$ and $a_i^\dagger a_j^\dagger$ satisfies
\be
X^{ij} = \tilde{X}^{ij} , \qquad Y^{[ij]}=0.
\ee
Its general element is then
\be\label{un}
A_{{\mathfrak u}(n)} =  X^{ij} (a_i a_j^\dagger + a_i^\dagger a_j) - \im Y^{ij} (a_i a_j^\dagger - a_i^\dagger a_j) ,
\ee
which is anti-Hermitian. 

The described ${\mathfrak u}(n)$ subalgebra is the one that in the vector representation acting on $\R^{2n}$ is compatible with the complex structure $J$ chosen. To check this we just need to verify that the action of (\ref{un}) on $\R^{2n}$ preserves the eigenspaces of $J$. The vector representation of $\mathfrak{spin}(2n)$ arises by considering the commutator of a general Lie algebra element $A(X)$ with a general linear combination of the $\Gamma$-matrices
\be
[A(X), y^A \Gamma_A ] :=  ( X  \vartriangleright y)^A \Gamma_A.
\ee
One of the two eigenspaces of $J$ is spanned by vectors $ y^i a_i, y^i\in \C$. Using
\be
[ a_i a_j^\dagger, a_k]= 2 \delta_{jk} a_i, \qquad [ a_i^\dagger a_j, a_k] = -2\delta_{ik} a_j, \qquad [a_i^\dagger a_j^\dagger, a_k] = 2\delta_{jk} a_i^\dagger - 2\delta_{ik} a_j^\dagger,
\ee
it is easy to see that $[A(X), y^i a_i]\in {\rm Span}(a_i)$ if and only if $A(X)\in {\mathfrak u}(n)$. 

\subsection{Pure spinors and complex structures}

The developed creation/annihilation operator model of $\mathfrak{spin}(2n)$ also comes with a preferred spinor. Indeed, we have the spinor given by the wedge product $e_1\wedge \ldots\wedge e_n$ of all $e_i$, which is the top degree polyform. The stabiliser of this spinor is $\mathfrak{su}(n)$. 

The spinor $e_1\wedge \ldots\wedge e_n$ is annihilated by all creation operators $a_i$, and so the dimension of the subspace of $\R^{2n}_\C$ that annihilates this spinor is $n$. This leads to the following definition. Let $\psi$ be a spinor. The vectors from $V$ act on $\psi$ by the Clifford multiplication. Denote by $V^0_\psi\subset V$ the subspace that annihilates $\psi$. The dimension of this subspace can be shown to be ${\rm dim}V^0_\psi \leq n$. If this dimension is maximal, i.e. ${\rm dim}(V^0_\psi)=n$, then $\psi$ is said to be a {\bf pure (or simple) spinor}. Thus our model comes with a preferred pure spinor $e_1\wedge \ldots\wedge e_n$.
 
Depending on $n$, the spinor $e_1\wedge \ldots\wedge e_n$ is either in $S^+$ or $S^-$. Another preferred pure spinor is the "identity" polyform $\id\in \Lambda^0(\C^n)$. It is an element of $S^+$. The result of the action of ${\rm Spin}(2n)$ on a pure spinor is a pure spinor. It can be shown that ${\rm Spin}(2n)$ acts transitively on the orbits of pure spinors in both $S^\pm$, see e.g. \cite{Harvey} for a proof. 

We have seen that choosing a complex structure $J$ on $\R^{2n}$ gives rise to the creation/annihilation operator model of ${\rm Cliff}_{2n}$, and to a preferred pure spinor $e_1\wedge \ldots\wedge e_n$. This correspondence works in the opposite direction as well. Thus, each pure spinor $\psi$ defines a complex structure on $\R^{2n}$. Explicitly, this complex structure can be computed by first computing the real 2-form
\be\label{moment-map}
M_{AB}:=\langle \hat{\psi}, \Gamma_A \Gamma_B \psi\rangle.
\ee
Here $\hat{\psi}$ is the result of the action of the appropriate charge conjugation operator that is constructed from either $R,R'$ anti-linear operators discussed in subsection \ref{sec:RR'}. When $\psi$ is a pure spinor, the matrix $M_A{}^B$ with one of its indices raised with the metric on $\R^{2n}$ squares to a multiple of the identity, and its appropriate multiple is then the sought complex structure. We will illustrate this general construction on examples below. 

The definition in (\ref{moment-map}) makes sense for any $n$. Indeed, when $n$ is even, the inner product is a pairing $\langle S^+, S^+\rangle$,  $\langle S^-, S^-\rangle$, because the top form is an even form in this case, and so even degree forms pair to even degree forms (and odd to odd). The action of two $\Gamma$-matrices on $\psi$ does not change the parity of the differential form. We also have the fact that both $R,R'$ are given by the product of an even number of $\Gamma$-matrices when $n$ is even, and so $\hat{\psi}$ is of the same parity as $\psi$. The pairing as in (\ref{moment-map}) is then possible. For $n$ odd the consideration is similar except that in this case $\hat{\psi}$ is of the parity opposite to that of $\psi$, and $\hat{\psi}, \psi$ can again be paired via the inner product. 

\section{Representations of Spin($n,n$)}
\label{sec:n-n}

\subsection{Clifford algebra, Lie algebra, inner product}

The creation/annihilation operator mode for ${\rm Cliff}(n,n)$ works similarly to the already treated case of ${\rm Cliff}(2n)$. The main difference is that there is now no need to introduce factors of the imaginary unit into the definition of the $\Gamma$-matrices. The $\Gamma$-matrices generating ${\rm Cliff}(n,n)$ are real linear combinations of the creation/annihilation operators, and they act on polyforms with real coefficients. 

We note that while the real model that we describe in this Section is canonical, there are other possible creation/annihilation operator models that are available even for the split case ${\rm Cliff}(n,n)$. These will be described after we understand the possible types of maximally-isotropic subspaces of $\R^{r,s}$ in the next Section. 

The model proceeds by selecting a pair $E^\pm$ of maximally-isotropic subspaces $E^\pm \sim \R^n$ that span $\R^{n,n}$. We will return to the geometry involved in such a choice below. For now, we assume that such a choice has been made, and consider the space $\Lambda(\R^n)$ of differential forms on $\R^n$ with real coefficients. We again define the creation/annihilation operators $b_i, b_i^\dagger$ as in (\ref{cr-an}). We denoted the creation/annihilation operators acting on $\Lambda(\R^n)$ by a different letter from those acting on $\Lambda(\C^n)$ because in the following section we are going to mix these two types of operators, and it helps to use different letters to keep track of which operator does what. The $\Gamma$-operators are then defined as follows
\begin{equation}
    \Gamma_{i}  \coloneqq b_i+b_i^\dagger ,\qquad 
    \Gamma_{n+i} \coloneqq b_i - b_i^\dagger, \qquad
    1 \leq i \leq n.
\end{equation}
it is easy to check that they satisfy the following Clifford algebra relations
\begin{equation}
    \Gamma_A\Gamma_B+\Gamma_B\Gamma_A=2\eta_{AB}\textbf{1}, \qquad 1 \leq A,B \leq 2n
\end{equation}
Where $\eta={\rm diag}(+1, \ldots, +1, -1, \ldots, -1)$. The directions $\Gamma_i + \Gamma_{i+n}$ span $E^+$, while $\Gamma_i + \Gamma_{i+n}$ span $E^-$. 

The Lie algebra $\mathfrak{spin}(n,n)$ is generated by products of distinct $\Gamma$-matrices, as in (\ref{spin-2n}). As in the $\mathfrak{spin}(2n)$ case the action of $\mathfrak{spin}(n,n)$ preserves the space of even and odd polyforms, and so the space $S=\Lambda(\R^n)$ splits $S=S^+\oplus S^-$ into the spaces of even and odd polyforms on which the Lie algebra $\mathfrak{spin}(2n)$ acts irreducibly. 

The $\mathfrak{spin}(n,n)$-invariant inner product on $S$ is still given by (\ref{inner-prod}). There are no non-trivial reality condition operators that can be constructed in the split signature case. Indeed, there are no imaginary $\Gamma$-matrices, and so $R'$ that was constructed as the product of all the imaginary $\Gamma$-matrices followed by the complex conjugation is just the complex conjugation. We have already required all polyforms to be real, and so $R'$ acts trivially. The product of all the real $\Gamma$-matrices followed by the complex conjugation then becomes just the product of all the $\Gamma$-matrices
\be
R = \Gamma_1 \ldots \Gamma_{2n}.
\ee
This operator squares to the identity $R^2=\id$, and its eigenspaces are the spaces $S^\pm$ of even and odd polyforms. 

\subsection{Preferred $\mathfrak{gl}(n)$ subalgebra}

The general $\mathfrak{spin}(n,n)$ Lie algebra element can be written as
\be\label{spin-nn}
A_{{\mathfrak spin}(n,n)} = \frac{1}{2} X^{ij} (b_i + b_i^\dagger) (b_j + b_j^\dagger) + \frac{1}{2} \tilde{X}^{ij}  (b_i - b_i^\dagger) (b_j -b_j^\dagger) + Y^{ij} (b_i + b_i^\dagger) (b_j - b_j^\dagger).
\ee
There is a preferred subalgebra of transformations that don't mix the polyforms of different degrees. This is generated by the product of a creation and an annihilation operators. Such transformations satisfy
\be
X^{ij} + \tilde{X}^{ij} =0, \qquad Y^{[ij]}=0,
\ee
and are of the form
\be\label{gl-n}
A_{{\mathfrak gl}(n)} = 2(X^{ij}-Y^{ij})  b_i b_j^\dagger  +2 Y^{ij} \delta_{ij} \id,
\ee
where $X^{ij}$ is anti-symmetric and $Y^{ij}$ is symmetric. This generates a $\mathfrak{gl}(n)$ subalgebra. 

\subsection{Pure spinors}

We can now describe the geometry involved in the choice of a pair of maximal totally isotropic subspaces $\R^n$ of $\R^{n,n}$ and thus the described model of ${\rm Cliff}(n,n)$. The novelty as compared to the case of ${\rm Cliff}(2n)$ is that choosing one such a maximal isotropic subspace does not uniquely define its complement in $\R^{n,n}$. 

As in the case ${\rm Cliff}(2n)$, the choice of a model can be encoded into a geometric structure. In the case of ${\rm Cliff}(2n)$ the geometric structure was a complex structure on $\R^{2n}$ that provided the split of the complexification $\R^{2n}_\C$ into two maximally isotropic subspaces $\C^n, \overline{\C^n}$. In the split signature case $\R^{n,n}$, the analog of this is a choice of a {\bf paracomplex} structure $I\in {\rm End}(\R^{n,n})$. This is an operator that squares to plus the identity $I^2=+\id$, so that its eigenspaces of eigenvalue $\pm 1$ are real. This operator is also compatible with the split signature metric, but the compatibility condition now involves a sign
\be\label{PC-compat}
\eta( I\cdot, I\cdot) = - \eta(\cdot,\cdot).
\ee
As the consequence of this extra minus sign, the eigenspaces of $I$ are totally null. Indeed, if $u,v\in E^+$, where $E^\pm:=\{ v\in \R^{n,n}: Iv=\pm v\}$, then $\eta(u,v) = - \eta(Iu,Iv)=-\eta(u,v)$, and so $E^+$ is totally isotropic, and of dimension $n$, and thus maximally totally isotropic. The same holds for $E^-$. Thus, choosing a metric-compatible paracomplex structure $I$ provides a decomposition $\R^{n,n}=E^+\oplus E^-$ into the two maximal isotropic subspaces. 

At this level the story is analogous to that for $\R^{2n}$ and ${\rm Cl}(2n)$. The novelty arises because in $\R^{n,n}$ a choice of a maximal totally isotropic subspace $E^+$ does not define $E^-$. Thus, a choice of only $E^+$ is not equivalent to a choice of a paracomplex structure $I$. The latter carries more information than the former. And it is only $E^+$ that is in correspondence with pure spinors as we now discuss.

Similarly to the case of ${\rm Cl}(2n)$, the described creation/annihilation operator model of ${\rm Cl}(n,n)$ comes with a preferred spinor given by $e_1\wedge \ldots \wedge e_n$. This spinor is annihilated by all the creation operators, and so the subspace of ${\rm Cl}(n,n)$ that annihilates it is the maximal totally isotropic subspace ${\rm Span}(\Gamma_i + \Gamma_{i+n})$. Thus, $e_1\wedge \ldots \wedge e_n$ is a pure spinor. In the case of ${\rm Cl}(2n)$ the stabiliser of this pure spinor is $\mathfrak{su}(n)$. It is clear that the analogous subgroup in the case of ${\rm Cl}(n,n)$ is $\mathfrak{sl}(n)$, and indeed it is easy to see that $e_1\wedge \ldots \wedge e_n$ is stabilised by $\mathfrak{sl}(n)$ as in (\ref{gl-n}) with $Y^{ij}\delta_{ij}=0$. The difference with the ${\rm Cl}(2n)$ case is that the stabiliser of $e_1\wedge \ldots \wedge e_n$ is larger than $\mathfrak{sl}(n)$. Indeed, it is clear that $e_1\wedge \ldots \wedge e_n$ is also killed by all transformations (\ref{spin-nn}) with 
\be
Y^{ij} = \frac{1}{2}( X^{ij} + \tilde{X}^{ij}),
\ee
as these transformations involve the product of two copies of creation operators, and thus kill the pure spinor $e_1\wedge \ldots \wedge e_n$. Thus, the stabiliser algebra of the pure spinor is the sum $\mathfrak{sl}(n) \oplus N$, where $N$ is a nilpotent subalgebra of dimension $n(n-1)/2$. The subalgebra $N$ is what contains the so-called two-form transformations that are very important in the generalised geometry context \cite{Hitchin:2003cxu}.

There is another natural pure spinor that comes with the model, this is the even spinor given by the identity polyform $\id \in \Lambda^0(\R^n)$. It is annihilated by all the annihilation operators, and thus its null subspace in ${\rm Cl}(n,n)$ is ${\rm Span}(\Gamma_i - \Gamma_{i+n})$. 

In the case of ${\rm Spin}(2n)$ a choice of a pure spinor is in one-to-one correspondence with a choice of a complex structure $J$ on $\R^{2n}$. There is no analogous statement in the case of ${\rm Spin}(n,n)$, as the choice of a pure spinor just selects a maximal totally isotropic subspace. This is not sufficient to define its complement in $\R^{n,n}$, and thus not sufficient to define a paracomplex structure. But the discussion above shows that it maybe possible to define a paracomplex structure by selecting two pure spinors, each of which defining its own maximally isotropic subspace. Indeed, we have seen that in our model the two pure spinors $\id\in S^+$ and $e_1\wedge \ldots \wedge e_n$ together define both $E^\pm$. One can expect this statement to generalise, so that if $\psi_{1,2}$ are pure spinors such that $\langle \psi_1,\psi_2\rangle \not=0$ then
\be
M_{AB} := \langle \psi_1, \Gamma_A \Gamma_B \psi_2\rangle
\ee
gives the operator $M_A{}^B$ whose square is a multiple of the identity operator and can thus be used to define a paracomplex structure $I$. We will illustrate this construction on specific examples below. 

We summarise this subsection by saying that the creation/annihilation operator model for ${\rm Cl}(n,n)$ arises when a paracomplex structure on $\R^{n,n}$ is chosen. The main difference with the ${\rm Cl}(2n)$ case is that a choice of a pure spinor is no longer equivalent to a choice of such a structure. A pure spinor defines a maximally isotropic subspace of $\R^{n,n}$, while the paracomplex structure defines two such complementary subspaces. As the result, a paracomplex structure can only be equivalent to a suitable pair of pure spinors. 

Pure spinors of the split signature pseudo-orthogonal spin groups ${\rm Spin}(n,n)$ are of relevance in the context of generalised geometry, see e.g. \cite{Hitchin:2003cxu}, as defining maximally isotropic subspaces of $\R^{n,n}$. 

\section{Pure spinors and maximally-isotropic subspaces}
\label{sec:MTN}

The purpose of this section is to review, in the amount we need, results of \cite{KT} on maximally-isotropic subspaces of $\R^{r,s}$. From now on we shall adopt the same terminology as in \cite{KT}, and refer to a maximally-isotropic subspace as {\bf MTN}, which stands for the maximal totally null.

\subsection{Pure spinors and MTN in the complex setting}

We now adopt some of the notation from \cite{KT}. As before, let $V$ be a real vector space equipped with a metric of signature $(r,s)$, with $r+s=2m$. As in \cite{KT}, we denote its complexification $V_\C=W$. Given a spinor $\psi$ we define $M(\psi):= W^0_\psi$, the subspace of the complexification $W$ of $V$ that annihilates $\psi$ via the Clifford product. The spinor $\psi$ is said to be pure if the dimension of the space $M(\psi)$ is maximal possible, i.e. $m$. It is known since Cartan \cite{Cartan} that pure spinors are Weyl. Cartan also gives a very useful algebraic characterisation of pure spinors.

Given a (Weyl) spinor $\psi$ (or a pair of Weyl spinors $\psi,\phi$), one can insert a number of $\Gamma$-matrices between two copies of $\psi$ (or, more generally, between $\psi$ and $\phi$) Let us introduce the convenient notation
\be
\Lambda^k(\R^{r,s})\ni B_k(\psi,\phi) := \langle \psi, \underbrace{\Gamma\ldots \Gamma}_{\text{$k$ times}} \phi\rangle.
\ee
It is assumed that distinct $\Gamma$-matrices are inserted, thus giving components of a degree $k$ differential form (anti-symmetric tensor) in $\R^{r,s}$. 

The following proposition is due to Cartan \cite{Cartan}:
\begin{theorem} \label{thm:cartan} If $\psi$ is a pure (simple) spinor, then $B_k(\psi,\psi) =0$ for $k\not=m$ and the $m$-vector $B_m(\psi,\psi)$ is proportional to the wedge product of the vectors constituting a basis of $M(\psi)$, where $M(\psi)$ is the MTN that corresponds to $\psi$. 
\end{theorem}
We note that this theorem gives a practical way of recovering $M(\psi)$ as the set of vectors whose insertion into $B_m(\psi,\psi)$ vanishes. Thus, this theorem establishes a one-to-one correspondence between pure spinors $\psi$ and MTN subspaces $M(\psi)$. This theorem also gives a set of quadratic constraints that each pure spinor must satisfy.

\subsection{The real index of an MTN subspace}

The complexification $W$ of $V$ is complex $2m$ dimensional, and the maximal dimension of a null (isotropic) subspace is $m$. The space of maximal, i.e. $m$-dimensional totally null (MTN) subspaces has the structure of a complex $m(m-1)/2$ dimensional manifold diffeomorphic to ${\rm O}(2m,\C)/{\rm U}(m)$. 

Given an MTN subspace $N\subset W$, one can consider the space $N\cap V$ of real vectors in $N$. The dimension of this subspace of real null vectors in $N$ is called the {\bf real index} of $N$. It is clear that for $\R^{r,s}$ with $r\geq s$ the real index can be as large as $s$. At the same time, in the case $(2\rho,2\sigma)$ the real index of MTN subspace can be as small as zero, while in the case $(2\rho+1,2\sigma+1)$ the minimal value of the real index is one. 

We have the following theorem from \cite{KT}: 
\begin{theorem}
The group ${\rm SO}(r,s)$ acts transitively on each set of all MTN subspaces of $W$ with a given real index and a given helicity. 
\end{theorem}
The notion of helicity of an MTN subspace arises because, as we reviewed above, there is a natural one-to-one correspondence between directions of pure (simple) spinors and MTN subspaces. Pure spinors are Weyl, and so of a given helicity. 

\subsection{Quadratic constraints in the real setting}

We have a quadratic constraint statement in the real setting, due to \cite{KT}:
\begin{theorem} The algebraic constraints for a simple spinor $\psi$ to have the real index equal to $r$ are
\be
B_r(\hat{\psi},\psi) \not=0, \qquad B_{r-2p}(\hat{\psi},\psi) = 0.
\ee
Here hat denotes a suitable charge conjugation, see Section III of \cite{KT}. The non-vanishing multi-vectors are $B_r(\hat{\psi},\psi)$, which is proportional to the wedge product of the real basis vectors in $M(\psi)$, as well as $B_{r+2p}(\hat{\psi},\psi)$, proportional to $B_r(\hat{\psi},\psi)$ wedged with $p$ copies of the K\"ahler bivector $j$. The K\"ahler bivector $j$ is given by the sum 
\be
j \sim \im \sum m\wedge \bar{m},
\ee
where $m$ is a suitably-chosen (orthonormal) basis of the complex vectors in $M(\psi)$.
\end{theorem}

\section{Models of Cl($r,s$)}
\label{sec:gen}

The idea now is to combine the previously described models of ${\rm Cl}(2n)$ and ${\rm Cl}(n,n)$ into a model for ${\rm Cl}(r,s)$. The key point of what follows is that there is not one resulting model of ${\rm Cl}(r,s)$, but in general many. There is one model corresponding to each choice of an MTN in (the complexification of the) $\R^{r,s}$, as well as a complement MTN. Given that there in general several different possible types of MTN that can be chosen (of different real index), we obtain a sequence of modes for each ${\rm Cl}(r,s)$. Each such model has a preferred "vacuum" state, or actually a pair of such preferred states, both are pure spinors. The real index of this "vacuum" state pure spinor is the real index of the MTN chosen in constructing a model. There is always a preferred such model corresponding to an MTN of the largest possible real index, which with $r\geq s$ is $s$. For ${\rm Cl}(n,n)$ this maximal index model is the real model described in Section \ref{sec:n-n}. However, other models are possible, even for the split case ${\rm Cl}(n,n)$. The only case that can be described by a unique creation/annihilation operator model is ${\rm Cl}(2n)$.

\subsection{The maximal index model: Clifford algebra, Lie algebra, inner product}

We require that $r+s$ is even because the case of $r+s$ is odd is closely related to the $r+s$ even, and can be obtained from the latter. For applications we care about we only require the knowledge of the spin group (for most applications only the spin Lie algebra). There is then no distinction between ${\rm Spin}(r,s)$ and ${\rm Spin}(s,r)$. Therefore, without loss of generality we can assume $r\geq s$. We can then write
\be
\R^{r,s} = \R^{2n, 0} \oplus \R^{s,s}, \qquad n := (r-s)/2.
\ee
There is of course some choice in splitting $\R^{r,s}$ in this way, and we discuss the geometry involved in this choice later. 

We can now consider a mix of the of the creation/annihilation constructions on $\Lambda(\mathbb{C}^n)$ and $\Lambda\mathbb{R}^{s}$. We introduce creation/annihilation operators $a_i, a_i^\dagger, i=1,\ldots, n$ as those acting on $\Lambda(\C^n)$. We introduce creation/annihilation operators $b_I, b_I^\dagger, I=1,\ldots, s$ as those acting on $\Lambda(\R^s)$. The Clifford generators then arise as operators on $\Lambda(\mathbb{C}^n\oplus\mathbb{R}^{s})$ 
\be
           \Gamma_{i}  &\coloneqq a_i  +a_i^\dagger ,\qquad 
        \Gamma_{i+n}  &\coloneqq \im (a_i - a_i^\dagger) , \qquad i=1,\ldots, n
         \\ 
        \Gamma_{I+2n} &\coloneqq b_I + b_I^\dagger,\qquad
        \Gamma_{I+2n+s}  &\coloneqq b_I - b_I^\dagger, \qquad I=1,\ldots, s.
   \ee
These Gamma matrices satisfy the Clifford algebra relations
\begin{equation}
    \Gamma_A\Gamma_B+\Gamma_B\Gamma_A=2g_{AB}\textbf{1}, \qquad 1 \leq A,B \leq 2n,
\end{equation}
where $g={\rm diag}(\underbrace{ +1,\ldots, +1}_{\text{$2n+s$ times}}, \underbrace{-1,\ldots, -1}_{\text{ $s$ times}})$.

The Lie algebra is again generated by all products of pairs of distinct $\Gamma$-matrices. Lie algebra acts on spinors, which are elements of the space of all polyforms $S=\Lambda(\mathbb{C}^n\oplus\mathbb{R}^{s})$. This splits into the subspaces of even and odd polyforms $S=S^+\oplus S^-$. 
The inner product (\ref{inner-prod}) is still an invariant inner product on $S$. 

\subsection{Reality conditions, Majorana spinors}

As in the case of ${\rm Cl}(2n)$, the are only two anti-linear operators (up to a complex multiple) that either commute or anti-commute with all $\Gamma$-matrices. These operators are obtained by taking the product of all real operators followed by the complex conjugation, or of all imaginary operators again followed by the complex conjugation. Thus, we define
\begin{equation}
    R=\underbrace{\Gamma_{1} \ldots \Gamma_n}_{\text{$n$ factors}} \underbrace{\ldots}_{\text{$n$ factors omitted}} \underbrace{\Gamma_{2n+1} \ldots \Gamma_{2n+2s}}_{\text{$2s$ factors}} \ast, \qquad R'=\underbrace{\Gamma_{n+1}\ldots\Gamma_{2n}}_{\text{$n$ factors}} \ast
\end{equation}
The commutativity properties of these maps are summarised in the lemma
\begin{lemma}
    \begin{equation}
    R\Gamma_A=(-1)^{n-1}\Gamma_A R, \qquad R'\Gamma_A=(-1)^{n}\Gamma_A R', \qquad A \in \{1, \ldots, 2(n+s)\}
\end{equation}
\end{lemma}
which is $s$ independent. The squares of these maps are captured by the following lemma
\begin{lemma} \label{Lemma 8.1}
    \begin{equation}
        R^2=(-1)^{\frac{n(n-1)}{2}}\textbf{1}, \qquad (R')^2=(-1)^{\frac{n(n+1)}{2}} \textbf{1},
    \end{equation}
\end{lemma}
which is again $s$-independent. Thus, the existence of the reality conditions and Majorana spinors depends only on the number of complex directions in $\C^n \otimes \R^s$. There are no Majorana spinors when $n$ is even but not a multiple of four. There are Majorana spinors when $n$ is odd, and Majorana-Weyl spinors when $n\in 4\Z$. 

\subsection{A preferred subalgebra of $\mathfrak{spin}(r,s)$}

We now look for a subalgebra that does not mix polyforms of different degrees. To this end, it is useful to write a general Lie algebra element in terms of the creation/annihilation operators $a,b$. We have
\be\label{lie-rs}
A_{{\mathfrak spin}(r,s)} = \frac{1}{2} X^{ij} (a_i + a_i^\dagger) (a_j + a_j^\dagger) - \frac{1}{2} \tilde{X}^{ij}  (a_i - a_i^\dagger) (a_j -a_j^\dagger) +\im Y^{ij} (a_i + a_i^\dagger) (a_j - a_j^\dagger) 
\\ \nonumber
+\frac{1}{2} X^{IJ} (b_I+ b_I^\dagger) (b_J + b_J^\dagger) + \frac{1}{2} \tilde{X}^{IJ}  (b_I - b_I^\dagger) (b_J-b_J^\dagger) + Y^{IJ} (b_I + b_I^\dagger) (b_J - b_J^\dagger)  \\ \nonumber
+Z_{++}^{iJ} (a_i + a_i^\dagger)  (b_I+ b_I^\dagger) + Z_{+-}^{iI} (a_i + a_i^\dagger)  (b_I- b_I^\dagger)
+ \im Z_{-+}^{iI} (a_i - a_i^\dagger)  (b_I+ b_I^\dagger) + \im Z_{--}^{iI} (a_i - a_i^\dagger)  (b_I- b_I^\dagger) .
\ee
Here all parameters $X,Y,Z$ are real. The conditions that selects only the operators containing both a creation and annihilation operator are 
\be
X^{ij}=\tilde{X}^{ij}, \qquad Y^{[ij]}=0, \\ \nonumber
X^{IJ}+\tilde{X}^{IJ}=0, \qquad Y^{[IJ]} =0, \\ \nonumber
Z^{iI}_{++}=Z^{iI}_{+-}=Z^{iI}_{-+}=Z^{iI}_{--}=0.
\ee
This selects a subalgebra ${\mathfrak u}(n)\oplus\mathfrak{gl}(s)$. 

\subsection{Pure spinors in the maximal index case}

Given the described model, we have a preferred spinor, which is the top degree polyform in $\Lambda(\C^n \oplus \R^s)$. Its annihilator in $\R^{r,s}$ (complexified) acting by Clifford multiplication has dimension $m$, and so is maximal. Therefore this is a pure (simple) spinor. 

It is interesting to compute the stabiliser of this pure spinor. Its stabiliser subalgebra does not contains terms from (\ref{lie-rs}) that are built from a pair of annihilation operators, but contains terms with a pair of creation operators, or with a creation and annihilation operator. The last group of terms must be constrained to annihilate the pure spinor. The terms from the first line in (\ref{lie-rs}) are ones acting solely on $\Lambda(\C^n)$. The subset of them that kills the pure spinor $e_1\wedge \ldots \wedge e_n$ is $\mathfrak{su}(n)$. The surviving terms in the second line in (\ref{lie-rs}) are ones generating $\mathfrak{sl}(s)$ plus $s(s-1)/2$ terms satisfying
\be
Y^{IJ} = \frac{1}{2}(X^{IJ} +\tilde{X}^{IJ}).
\ee
For the last line in (\ref{lie-rs}), the conditions that there are no terms containing a pair of annihilation operators are
\be
Z^{iI}_{++}=Z^{iI}_{+-}, \qquad Z^{iI}_{-+}=Z^{iI}_{--}.
\ee
There are thus $2ns$ real such terms. The stabiliser subalgebra is then $\mathfrak{su}(n)\oplus\mathfrak{sl}(s)$ plus $s(s-1)/2+ 2ns$ generators. Its general element can be written as
\be
A_{stab} = 2(X^{ij} -\im Y_s^{ij} ) a_i a_j^\dagger + (X^{IJ} -\tilde{X}^{IJ}- 2Y_s^{IJ}) b_I b_J^\dagger  + (X^{IJ}+\tilde{X}^{IJ})b_I b_J    \\ \nonumber
+(Z_{+}^{iJ} +\im Z_{-}^{iI}) a_i b_I + (Z_{+}^{iJ} -\im Z_{-}^{iI}) a_i^\dagger b_I,
\ee
where $Y_s^{ij}, Y_s^{IJ}$ are the symmetric parts of $Y$'s and must be tracefree $Y_s^{ij}\delta_{ij}=0, Y_s^{IJ} \delta_{IJ}=0$. The quantities $Z^{iI}_{\pm}$ are $2ns$ real quantities. 

\subsection{Real index zero}

We now consider the models that become possible when $MTN$ of the real index that is not maximal is chosen. To avoid overcomplicating the notation, we will only treat in general the case of the minimal real index. 

We start by considering the case ${\rm Cl}(2\rho, 2\sigma)$. In this case the minimal possible real index is zero. This means that we represent
\be
\R^{2\rho, 2\sigma}= \R^{2\rho} \oplus \R^{2\sigma},
\ee
and then choose a complex structure in both summands. The corresponding $-\im$ eigenvalue eigenspace is $\C^\rho \oplus \C^\sigma$, and spinors become realised as polyforms in $\Lambda(\C^\rho\oplus \C^\sigma)$. 

We again introduce two pairs of creation/annihilation operators $a_i, a_i^\dagger, i=1,\ldots, \rho$ and $\tilde{a}_I, \tilde{a}_I^\dagger, I=1,\ldots, \sigma$. We referred to the second set as $\tilde{a}$ rather than $b$ to reserve the name $b$ to operators that act on a number of copies of $\R$ rather than $\C$. The $\Gamma$-operators now become
\be
\Gamma_i = a_i+a_i^\dagger, \quad \Gamma_{i+\rho} = \im(a_i-a_i^\dagger), \\ \nonumber
\Gamma_{I+2\rho} = \im (\tilde{a}_I + \tilde{a}_I^\dagger), \quad \Gamma_{I+2\rho+\sigma} = \tilde{a}_I - \tilde{a}_I^\dagger.
\ee
Note that the placement of the imaginary unit is now opposite in the $\C^\sigma$ factor as compared to the $\C^\rho$ factor. This generates the correct Clifford algebra ${\rm Cl}(2\rho, 2\sigma)$.

The discussion of the Lie algebra, inner product and semi-spinors is unchanged to the previous cases. The only novelty is in the available anti-linear operators. Again, these arise as the product of either all real or all imaginary $\Gamma$-matrices followed by the complex conjugation. Their squares can be deduced using (\ref{R2}). Thus, we define
\be
R= \underbrace{\Gamma_1 \ldots \Gamma_\rho}_{\text{$\rho$ times}} \underbrace{\Gamma_{1+2\rho+\sigma} \ldots \Gamma_{2\rho+2\sigma}}_{\text{$\sigma$ times}} \ast, \quad
R'= \underbrace{\Gamma_{1+\rho} \ldots \Gamma_{2\rho}}_{\text{$\rho$ times}} \underbrace{\Gamma_{1+2\rho} \ldots \Gamma_{2\rho+\sigma}}_{\text{$\sigma$ times}} \ast.
\ee
We then have
\be
R^2 = (-1)^{\rho\sigma} (-1)^{\rho(\rho-1)/2} (-1)^{\sigma(\sigma+1)/2} , \qquad
(R')^2 = (-1)^{\rho\sigma} (-1)^{\rho(\rho+1)/2} (-1)^{\sigma(\sigma-1)/2}.
\ee
This can be rewritten as
\be
R^2 = (-1)^{(\rho-\sigma)(\rho-\sigma-1)/2}, \qquad (R')^2 = (-1)^{(\rho-\sigma)(\rho-\sigma+1)/2},
\ee
which shows that their properties are controlled only by $\rho-\sigma$. So, the availability of Majorana and Majorana-Weyl spinors depends only on the signature, and not on a model used. 

The pure spinor arising in this model is $e_1\wedge \ldots \wedge e_\rho\wedge e_{1+\rho} \wedge \ldots e_{\sigma+\rho}$. The general Lie algebra element can be written as
\be
A_{(2 \rho,2\sigma)} = \frac{1}{2} X^{ij} (a_i + a_i^\dagger) (a_j + a_j^\dagger) - \frac{1}{2} \tilde{X}^{ij}  (a_i - a_i^\dagger) (a_j -a_j^\dagger) +\im Y^{ij} (a_i + a_i^\dagger) (a_j - a_j^\dagger) 
\\ \nonumber
- \frac{1}{2} X^{IJ} (\tilde{a}_I + \tilde{a}_I^\dagger) (\tilde{a}_J + \tilde{a}_J^\dagger) + \frac{1}{2} \tilde{X}^{IJ}  (\tilde{a}_I - \tilde{a}_I^\dagger) (\tilde{a}_J -\tilde{a}_J^\dagger) +\im Y^{IJ} (\tilde{a}_I + \tilde{a}_I^\dagger) (\tilde{a}_J - \tilde{a}_J^\dagger) 
\\ \nonumber
+\im Z_{++}^{iJ} (a_i + a_i^\dagger)  (\tilde{a}_I+ \tilde{a}_I^\dagger) + Z_{+-}^{iI} (a_i + a_i^\dagger)  (\tilde{a}_I- \tilde{a}_I^\dagger)
- Z_{-+}^{iI} (a_i - a_i^\dagger)  (\tilde{a}_I+ \tilde{a}_I^\dagger) + \im Z_{--}^{iI} (a_i - a_i^\dagger)  (\tilde{a}_I- \tilde{a}_I^\dagger) .
\ee
The terms in the first two lines that kill the pure spinor form the subalgebra ${\mathfrak su}(\rho)\oplus{\mathfrak su}(\sigma)$. The terms in the last line that kill the pure spinor are those that do not have pairs of annihilation operators and thus must satisfy 
\be
Z^{iI}_{++}+Z^{iI}_{--} =0, \qquad Z^{iI}_{+-}-Z^{iI}_{-+} =0.
\ee
This gives $2\rho\sigma$ real generators. 

\subsection{Real index one}

In the case ${\rm Cl}(2\rho+1, 2\sigma+1)$ the minimal value of the real index is one. This corresponds to the split 
\be
\R^{2\rho+1, 2\sigma+1}= \R^{2\rho} \oplus \R^{2\sigma}  \oplus \R^{1,1}.
\ee
The corresponding MTN subspace is obtained by choosing a complex structure in the first two summands, and a paracomplex structure in the last one. The MTN subspace is then $\C^\rho\oplus\C^\sigma \oplus \R$. 

To generate the corresponding model for the Clifford algebra we proceed as in the previous subsection, but append one pair of real creation/annihilation operators $c,c^\dagger$. The $\Gamma$-matrices are then
\be
\Gamma_i = a_i+a_i^\dagger, \quad \Gamma_{i+\rho} = \im(a_i-a_i^\dagger), \\ \nonumber
\Gamma_{I+2\rho} = \im (\tilde{a}_I + \tilde{a}_I^\dagger), \quad \Gamma_{I+2\rho+\sigma} = \tilde{a}_I - \tilde{a}_I^\dagger, \\ \nonumber
\Gamma_{1+2\rho+2\sigma} = b+ b^\dagger, \qquad \Gamma_{2+2\rho+2\sigma} = b- b^\dagger.
\ee
Taking the products of distinct $\Gamma$-matrices we obtain the Lie algebra. The inner product construction is unchanged. Taking the products of all real and all imaginary $\Gamma$-matrices followed by the complex conjugation we generate the anti-linear operators from which reality conditions can be constructed. The subtleties arising are most clearly seen by considering specific examples, which will be considered in the following sections.

\subsection{Mixed structures}
\label{sec:comp-paracomp}

A complex structure in $\R^{2n}$ provides a decomposition $\R^{2n}_\C= E^+ \oplus E^-$, where both $E^\pm$ are totally null, and arise as the eigenspaces of the complex structure operator. A paracomplex structure on $\R^{n,n}$ is similarly a structure that gives a decomposition $\R^{n,n}= E^+ \oplus E^-$ with $E^\pm$ again totally null, but this time real. 

It is clear that models of ${\rm Cl}(r,s)$ we described rely on a structure that is an appropriate mix of complex and paracomplex structures. The purpose of this subsection is to describe such more general structures in geometric terms. Examples will be given in the following sections. 

We start with a description of what such a structure must do, and then formalise these requirements. First, the structure we are after must provide a decomposition 
\be\label{decomp-gen}
\R^{r,s} = \R^{2k,2l} \oplus \R^{m,m}.
\ee
Second, the structure must select a pair of complementary MTN $E^\pm$ in both $\R^{2k,2l}$ and $\R^{m,m}$. In other words, after a decomposition (\ref{decomp-gen}) is chosen, the structure must select a complex structure $J: J^2=-\id$ in $\R^{2k,2l}$ and a paracomplex structure $I: I^2=\id$ in $\R^{m,m}$. 

We now extend $I,J$ to act on the whole of $\R^{r,s}$, with 
\be
\R^{2k,2l}={\rm Ker}(I), \qquad \R^{m,m} = {\rm Ker}(J).
\ee
Since both $I,J$ act by projecting on their respective factors, we have $IJ=JI=0$. We can then form the complex linear combination of these two maps:
\be
K:= I + \im J.
\ee
This is a linear map on $\R^{r,s}_\C$. It has the property 
\be
K^2 = I^2 + \im (IJ+JI) - J^2 =\id_{\R^{m,m}} + \id_{\R^{2k,2l}} =\id,
\ee
and so behaves as a paracomplex structure on $\R^{r,s}$, apart from the fact that this map is complex-valued. Another property the constructed map has is 
\be\label{reality-K}
K \bar{K}=\bar{K} K.
\ee
This follows from $IJ=JI=0$. One way to rephrase this is to say that the complex map $K$ commutes with its complex conjugate $\bar{K}$. Another, better way is to say that $K\bar{K}=\bar{K}K=P$ is real. This provides a reality condition on the map. Note that because $K^2=\id$ we have also $P^2=\id$. But now $P$ is a real map and so defines the splitting of $\R^{r,s}$ into two factors (\ref{decomp-gen}) with the $\R^{m,m}, \R^{2k,2l}$ being the eigenspaces of eigenvalue $+1$ and $-1$ respectively.

We can also describe the compatibility between $K$ constructed and the metric on $\R^{r,s}$. We have
\be
g(K X, KY) = g ((I+\im J) X, (I+\im J)Y) = g(IX,IY) - g(JX,JY).
\ee
No mixed terms of the type $g(IX,JY)$ arises because both $I,J$ project on one of the two factors in (\ref{decomp-gen}), and these factors are assumed metric-orthogonal in (\ref{decomp-gen}). We can now use the usual properties of paracomplex structure $I$ and complex structure $J$ acting on the two factors in (\ref{decomp-gen}). We have $g(IX,IY) = - g(X|_2, Y|_2)$, where $X|_2$ denotes the projection onto the second factor in (\ref{decomp-gen}). Similarly $g(JX,JY) = g(X|_1, Y|_1)$, where $X|_1$ denotes the projection onto the first factor in (\ref{decomp-gen}). This means that we have
\be
g(K X, KY) = - g(X|_2, Y|_2)- g(X|_1, Y|_1) = - g(X,Y).
\ee
Thus, the operator $K$ is metric-compatible in the same sense that a paracomplex structure is. 

The difference with the usual paracomplex structure is that $K$ is complex-valued, but satisfies the reality condition (\ref{reality-K}). It is given by a complex linear combination of a paracomplex and complex structures. We will refer to it as a structure of a mixed type or a {\bf mixed structure} for short. It is clear that as constructed, what this operator does is define an orthogonal decomposition (\ref{decomp-gen}) as kernels of its real and imaginary parts, as well as define a pair complex/paracomplex structure on the two factors in (\ref{decomp-gen}).

Having constructed an object with desired properties, let us axiomatise it:
\begin{definition} A {\bf mixed structure} on a real vector space $V$ equipped with a metric $g$ is a linear map $K: V_\C\to V_\C$ satisfying $K^2=\id, K\bar{K}=\bar{K}K$ and $g(KX,KY)=g(\bar{K}X,\bar{K}Y)=-g(X,Y)$. Here $V_\C$ is the complexification $V_\C=\C\otimes_\R V$ and $\bar{K}$ is the complex conjugate of $K$.
\end{definition} 

\begin{lemma} The eigenvalue $\pm 1$ eigenspaces $E^\pm \subset V_\C$ of $K$ are of the same complex dimension and are totally null. 
\end{lemma}
The nullness of $E^\pm$ follows straightforwardly from $K^2=\id$ and $g(KX,KY)=-g(X,Y)$. The fact that the dimensions are the same follows from the fact that the metric $g$ is only non-vanishing as a pairing between $E^+$ and $E^-$. All these facts are the same as for paracomplex structures, but in the complexified setting.

\begin{definition} A {\bf product} structure on a (real) vector space $V$ is a map $P\in {\rm End}(V)$ that squares to the identity $P^2=\id$. An {\bf orthogonal} product structure on $V$ equipped with a metric $g$ is one that satisfies $g(PX,PY)=g(X,Y)$. 
\end{definition}
Let $V^\pm$ be the eigenspaces of $P$ of eigenvalue $\pm 1$. Then $V^+$ is metric-orthogonal to $V^-$, and so $V=V^+\oplus V^-$. Indeed, we have the following simple calculation $g(X_+,Y_-)= -g(PX_+,PY_-)=g(X_+,Y_-)$, which means that this product is zero. This property justifies the name "product structure" for $P$. 

\begin{proposition} The operator $P:=K\bar{K}=\bar{K}K$, where $K$ is a complex/paracomplex structure, is an orthogonal product structure on $V$. 
\end{proposition}
Thus, a mixed structure can be thought of as a square root of a product structure. The proof is simple. Indeed, $P$ is a real operator, and so $P\in {\rm End}(V)$. It squares to the identity $P^2=K\bar{K} K\bar{K} = K\bar{K}\bar{K} K = \id$. Then, using $g(KX,KY)=g(\bar{K}X,\bar{K}Y)$, and applying it to $X=\bar{K} \tilde{X}, Y=\bar{K} \tilde{Y}$ we have $g(P\tilde{X},P\tilde{Y})= g(\tilde{X},\tilde{Y})$, and so $P$ is an orthogonal product structure. 

We now use $P=K\bar{K}$ to provide the orthogonal decomposition
\be
V= V^-_K \oplus V^+_K
\ee
into the eigenspaces of $P$. 
\begin{proposition} The operator $K$ acts on $V^+$ as a paracomplex structure, and on $V^-$ as $\im$ times a complex structure.
\end{proposition}
To prove this, we introduce the following notations
\be
K|_{V^+} := I, \qquad K|_{V^-} := \im J.
\ee
These notations are justified by the fact that both $I,J$ are real operators. Indeed, applying $K$ to a vector $X_+\in V^+$ we have $K X_+ = K K\bar{K} X_+= \bar{K} X_+$, and so $K$ acts on $V^+$ as a real operator $I:V^+\to V^+$. Similarly, acting on a vector $X_-\in V^-$ we have $K X_- = - K K\bar{K} X_- = - \bar{K} X_-$, which means that on $V^-$ the operator $K$ acts as an imaginary operator, or as $\im J$ with $J:V^-\to V^-$. 

There are some other useful properties that can be proven. We have
\be\label{K-preserves}
K(V^+)\subset V^+, \qquad K(V^-)\subset V^-.
\ee
Indeed, taking $\tilde{X}= KX_+$ and applying $P$ we have $P\tilde{X}= \bar{K} K KX_+= \bar{K} X_+= KX_+=\tilde{X}$, and so $\tilde{X}\in V^+$. Similarly, for $\tilde{X}= KX_-$ we have $P\tilde{X} = \bar{K} K K X_-= \bar{K} X_-=-KX_-=-\tilde{X}$. These properties mean that the maps $I,J$ are linear maps on $V^\pm$ respectively
\be
I:V^+\to V^+, \qquad J:V^-\to V^-.
\ee
It remains to show that
\be 
I^2= \id_{V^+}, \qquad J^2=-\id_{V^-},
\ee
where $\id_{V^\pm}$ are the projectors on $V^\pm$ respectively. This is easy. Indeed, we have $I^2= (K|_{V^+})^2 = K^2_{V^+}$ because of (\ref{K-preserves}). Therefore $I^2$ is the identity operator on $V^+$. Similarly, $(\im J)^2 =(K|_{V^-})^2 = K^2_{V^-}$, which again equals to the identity. Thus $J^2$ acts as minus the identity on $V^-$. The proposition is proven.

\section{Dimension two}
\label{sec:two}

The purpose of this and the following sections is to apply the constructions outlined above to the low-dimensional Clifford and Spin algebras. We treat the cases ${\rm Spin}(r,s), r+s\leq 6$. We only consider the cases $r+s=2m$, because the odd-dimensional cases can easily be obtained from the lower-dimensional even-dimensional case. We discuss the spinors, the inner product, as well as the possible reality conditions that can be imposed, i.e. Majorana spinors. Also, in some cases, we write down explicitly the arising Dirac operator. A useful companion to our description is the treatment of Section 2.4 "Orbits in the low dimensions" in \cite{Bryant}. 

Starting in dimension four it is possible to achieve an economy of description by introducing quaternions. This, and the corresponding octonionic construction is treated in the accompanying paper. 

\subsection{Spin(2)}

This has a polyform representation over $\Lambda(\mathbb{C})$. Let us denote the complex coordinate on $\C$ by $z$, and the basis vector in $\Lambda^1(\C)$ by $dz$. We have a generic polyform of the form
\begin{equation}
    \Psi=\alpha+\beta dz
\end{equation}
Where $\alpha,\beta \in \C$. The even $\alpha$ and odd $\beta dz$ parts here are the Weyl spinors. 
The ${\rm Cl}(2)$ is generated by the Gamma matrices
\begin{equation}
    \Gamma_1=a+a^{\dagger}, \quad \Gamma_2=i(a-a^{\dagger}).
\end{equation}
Their action on $\Psi$ is
\be
\Gamma_1(\alpha+\beta z) = \alpha dz + \beta, \qquad \Gamma_2(\alpha+\beta z) = \im \alpha dz -\im  \beta.
\ee
If we associate with $\Psi$ a 2-component column
\be
\Psi = \left(\begin{array}{c} \alpha \\ \beta \end{array}\right), 
\ee
the $\Gamma$-matrices take the following form
\be
\Gamma_1 = \left( \begin{array}{cc} 0 & 1 \\ 1 & 0 \end{array}\right), \qquad
\Gamma_2= \left( \begin{array}{cc} 0 & -\im \\ \im & 0 \end{array}\right).
\ee
The Lie algebra is generated by the product $\Gamma_1 \Gamma_2$
\begin{equation}
    \mathfrak{spin}(2)= 
    \bigg\{ 
    \begin{pmatrix}
    is&0\\
    0&-is
    \end{pmatrix}
    \bigg\vert
    \ s \in \mathbb{R}
    \bigg\}
    \sim \mathfrak{u}(1)
\end{equation}

The inner product (\ref{inner-prod}) takes the following form
\be\label{inner-2}
\langle \Psi_1, \Psi_2\rangle = ( \alpha_1+ \beta_1 dz) \wedge (\alpha_2 + \beta_2 dz) \Big|_{top} = \alpha_1 \beta_2+ \beta_1 \alpha_2.
\ee
This can be written in matrix form
\be
\langle \Psi_1, \Psi_2\rangle = \Psi_1^T \left( \begin{array}{cc} 0 & 1 \\ 1 & 0 \end{array}\right) \Psi_2.
\ee

There are two anti-linear operators
\be
R= \Gamma_1 \ast, \qquad R'=\Gamma_2 \ast,
\ee
with $R^2=\id, (R')^2=-\id$. Thus, we can use $R$ to impose the Majorana reality condition. We have
\be
R  \left(\begin{array}{c} \alpha \\ \beta \end{array}\right) =  \left(\begin{array}{c} \beta^* \\ \alpha^* \end{array}\right),
\ee
which means that Majorana spinors are of the form
\be
\Psi_M = \left(\begin{array}{c} \alpha \\ \alpha^* \end{array}\right).
\ee

The spinor $dz\in S_-$ is our canonical pure spinor associated with this model. It has a trivial stabiliser. The other canonical pure spinor is $\id\in S_+$. The generic Weyl spinors are multiples of these. 

It is also interesting to discuss the Dirac equation. We now allow $\alpha,\beta$ to become functions of the complex null coordinates $z,\bar{z}$ on $\R^2$. We take the usual relation 
\be
z = x_1+\im x_2,
\ee
so that the complex structure acts $J(x_1)=x_2, J(x_2)=-x_1$, and $J(z)=-\im z$. We have
\be
\frac{\partial}{\partial x_1} = \frac{\partial}{\partial z}+ \frac{\partial}{\partial \bar{z}}, \qquad
\frac{\partial}{\partial x_2} = \im \frac{\partial}{\partial z}-\im  \frac{\partial}{\partial \bar{z}}.
\ee
The Dirac operator is 
\be
D: = \Gamma^1 \frac{\partial}{\partial x_1} + \Gamma^2 \frac{\partial}{\partial x_2} = 2( a \frac{\partial}{\partial z}+ a^\dagger  \frac{\partial}{\partial \bar{z}}).
\ee
This makes it clear that 
\be
\Psi(z,\bar{z}) = \alpha(\bar{z})+ \beta(z)dz
\ee
is in the kernel of the Dirac operator. This is how solutions of the Cauchy-Riemann equations on the complex plane are the same as the solutions of the (massless) Dirac equation.

\subsection{Spin(1,1)}
This has a polyfrom representation over $\Lambda(\mathbb{R})\sim\mathbb{R}^2$ instead. If we use $u$ to denote the null coordinate and $du$ the corresponding one-form, the spinor polyform is now
\begin{equation}
    \Psi=\alpha+\beta du,
\end{equation}
where $\alpha,\beta \in \mathbb{R}$. The ${\rm Cl}(1,1)$ Gamma operators are
\begin{equation}
    \Gamma_1=a+ a^{\dagger}, \quad \Gamma_2=a- a^{\dagger}.
\end{equation}
In matrix notations, this corresponds to 
\be
\Gamma_1 = \left( \begin{array}{cc} 0 & 1 \\ 1 & 0 \end{array}\right), \qquad
\Gamma_2= \left( \begin{array}{cc} 0 & -1 \\ 1 & 0 \end{array}\right).
\ee
The Lie algebra is generated by $\Gamma_1 \Gamma_2$ 
\begin{equation}
    \mathfrak{spin}(1,1)= 
    \bigg\{ 
    \begin{pmatrix}
    k&0\\
    0&-k
    \end{pmatrix}
    \bigg\vert
    \ k \in \mathbb{R}
    \bigg\}
\end{equation}
The invariant inner product is still given by (\ref{inner-2}). There are no non-trivial reality conditions to be imposed, the spinors are explicitly real. The canonical pure spinors are $du\in S_-$ and $\id\in S_+$. Both have trivial stabilisers. 

The Dirac equation is also interesting. If $x,y$ are the coordinates on $\R^{1,1}$ so that $ds^2=dx^2-dy^2$, then we can take $u=x-y, v=x+y$ as the null coordinates, and 
\be
\frac{\partial}{\partial x} = \frac{\partial}{\partial u}+ \frac{\partial}{\partial v}, \qquad
\frac{\partial}{\partial y} =  \frac{\partial}{\partial u}-\frac{\partial}{\partial v}.
\ee
We then have
\be
D= \Gamma^1 \frac{\partial}{\partial x} +\Gamma_2 \frac{\partial}{\partial y} = 2( a \frac{\partial}{\partial u}+ a^\dagger  \frac{\partial}{\partial v}).
\ee
The Dirac spinor
\be
\psi(u,v) = \alpha(v) + \beta(u) du 
\ee
is in the kernel of the Dirac operator.

\section{Dimension four}
\label{sec:four}

Things become much more interesting in dimension four. There are three signatures to consider. The Euclidean, the Lorentzian and split. The Euclidean and Lorentzian cases have just one possible model each. In the case of the split signature there are two possible models, one corresponding to the real index equal to two, the other with real index zero. Thus, there are two types of pure spinors in the split case. 

\subsection{Spin(4)}

We choose a complex structure thus identifying $\R^4=\C^2$. We will call the arising null complex coordinates $z_{1,2}$, and the corresponding one-forms $dz_{1,2}$. We introduce two pairs of creation/annihilation operators $a_{1,2}, a_{1,2}^\dagger$. The $\Gamma$ operators take the following form
\begin{equation}\label{gamma-spin4}
    \begin{split}
        \Gamma_4&=a_1+ a_1^{\dagger},\\
        \Gamma_2&=a_2 + a_2^{\dagger} ,\\
    \end{split}
    \qquad
    \begin{split}
        \Gamma_3&=-i(a_1 - a_1^{\dagger} ),\\
        \Gamma_1&=-i(a_2 - a_2^{\dagger} ).\\
    \end{split}
\end{equation}
We have adopted the numbering and the signs in the imaginary $\Gamma$-matrices that become convenient below. A generic Dirac spinor (general polyform) is given by 
\begin{equation}\label{psi-spin4}
    \Psi=(\alpha+\beta dz_{12})+(\gamma dz_1+\delta dz_2),
\end{equation}
where $dz_{12}:=dz_1\wedge dz_2$ and $\alpha,\beta,\gamma,\delta\in\C$. In matrix notations, the Dirac spinor is 4-component. It is convenient to adopt the $2\times 2$ block notations, in which Weyl spinors are 2-component. Thus, we write
\be
\Psi = \left(\begin{array}{c} \psi_+ \\ \psi_- \end{array}\right), \qquad
\psi_+ = \left(\begin{array}{c} \alpha \\ \beta \end{array}\right), \qquad
\psi_- = \left(\begin{array}{c} \gamma \\ \delta \end{array}\right).
\ee
The action of the $\Gamma$ operators is as follows
\be
\Gamma_4 \Psi = \alpha dz_1 + \beta dz_2 + \gamma+ \delta dz_{12}, \\ \nonumber
\Gamma_3 \Psi = -\im \alpha dz_1 + \im \beta dz_2 + \im \gamma- \im \delta dz_{12}, \\ \nonumber
\Gamma_2 \Psi =  - \beta dz_1 +\alpha dz_2+ \delta- \gamma dz_{12}, \\ \nonumber
\Gamma_1 \Psi = -\im \beta dz_1 - \im \alpha dz_2 + \im \delta+\im \gamma dz_{12}.
\ee
In matrix notations this becomes
\be\label{gamma-matr-spin4}
\Gamma_4 =\left( \begin{array}{cc} 0 & \id \\ \id & 0 \end{array}\right), \quad 
\Gamma_i = \left( \begin{array}{cc} 0 & \im \sigma^i \\ -\im \sigma^i & 0 \end{array}\right), \quad i=1,2,3.
\ee
Here $\sigma^i$ are the usual Pauli matrices. It is this simple form of the resulting $\Gamma$-matrices that motivated the choices made in (\ref{gamma-spin4}), (\ref{psi-spin4}). 

The Lie algebra is generated by products of distinct $\Gamma$-matrices. This gives a $4\times 4$ Lie algebra matrix that is block-diagonal. Let us refer to its $2\times 2$ blocks as $A, A'$, where $A$ acts on $S_+$ and $A'$ on $S_-$ respectively. We have
\be
A = \im (- \omega^{4i} + \frac{1}{2} \epsilon^{ijk} \omega^{jk} )\sigma^i , \qquad
A' = \im ( \omega^{4i} + \frac{1}{2} \epsilon^{ijk} \omega^{jk} )\sigma^i.
\ee
Both are anti-Hermitian $2\times 2$ matrices. This demonstrates $\mathfrak{spin}(4)=\mathfrak{su}(2)\oplus\mathfrak{su}(2)$.  

The invariant inner product is determined by the following computation
\be
\langle \Psi_1, \Psi_2\rangle = ( \alpha_1 - \beta_1 dz_{12} + \gamma_1 dz_1 + \delta_1 dz_2) \wedge ( \alpha_2 + \beta_2 dz_{12} + \gamma_2 dz_1 + \delta_2 dz_2) \Big|_{top} = \\ \nonumber
(\alpha_1 \beta_2-\alpha_2\beta_1) + (\gamma_1 \delta_2 - \gamma_2 \beta_1). 
\ee
It is thus an anti-symmetric pairing $\langle S_+, S_+\rangle, \langle S_-,S_-\rangle$. It can be written in matrix terms as
\be\label{inner-4}
\langle \Psi_1, \Psi_2\rangle = \Psi_1^T \left( \begin{array}{cc} \epsilon & 0 \\ 0 & \epsilon \end{array}\right) \Psi_2, \qquad \epsilon:= \im \sigma^2 =  \left( \begin{array}{cc} 0 & 1 \\ -1 & 0 \end{array}\right).
\ee
For the possible reality conditions, both $R=\Gamma_2 \Gamma_4 \ast$ and $R'=\Gamma_1 \Gamma_3\ast$ square to minus the identity, and so there are no Majorana spinors in this case. Of them $R'$ commutes with all $\Gamma$-matrices and defines the hat operator 
\be
\hat{\psi}_+ = \left(\begin{array}{c} \alpha \\ \beta \end{array}\right)^\wedge =\left(\begin{array}{c} -\beta^* \\ \alpha^* \end{array}\right) , \qquad
\hat{\psi}_-=\left(\begin{array}{c} \gamma \\ \delta \end{array}\right)^\wedge = \left(\begin{array}{c} \delta^* \\ -\gamma^* \end{array}\right),
\ee
which squares to minus the identity. 

There are two "canonical" pure spinors that come with the model, "identity" spinor $\id$ and the top polyform $dz_{12}$. They are both in $S_+$. The stabiliser of both is the copy of $\mathfrak{su}(2)\subset \mathfrak{spin}(4)$ whose action on $S_+$ is trivial. 

It is clear that a generic Weyl spinor of ${\rm Spin}(4)$ is also pure, with stabiliser ${\rm SU}(2)$. The group ${\rm Spin}(4)$ acts transitively on the space of Weyl spinors of fixed norm $\langle \hat{\psi}_+, \psi_+\rangle$. This space is the 3-sphere $S^3$. 

In the case of ${\rm Cl}(2n)$, pure spinors are in one-to-one correspondence with complex structures. The complex structure on $\R^4$ corresponding to a generic Weyl spinor can be recovered as in (\ref{moment-map}). Let us consider the case of a spinor in $S_+$. A simple computation gives
\be
\langle \hat{\psi}_+, \Gamma\Gamma \psi_+\rangle =  \im \Sigma^i V^i, 
\ee
where
\be
\Sigma^i = dx^4 \wedge dx^i - \frac{1}{2} \epsilon^{ijk} dx^j \wedge dx^k
\ee
is the basis of self-dual 2-forms and
\be
V^i := {\rm Tr}\left( \psi_+^\dagger \sigma^i \psi_+\right) = ( 2 {\rm Re}(\alpha^*\beta),   2 {\rm Im}(\alpha^*\beta), |\alpha|^2-|\beta|^2)
\ee
is a 3-vector with squared norm 
\be
|V|^2=V^i V^i = (|\alpha|^2+|\beta|^2)^2 = \langle \hat{\psi}_+, \psi_+\rangle^2.
\ee
If one raises an index of $\Sigma^i_{\mu\nu}$, one obtains a triple of endomorphisms of $\R^{4}$ that satisfy the algebra of the quaternions
\be
\Sigma^i_{\mu}{}^{\rho} \Sigma^j_\rho{}^\nu = - \delta^{ij} \delta_\mu{}^\nu + \epsilon^{ijk} \Sigma^k_\mu{}^\nu.
\ee
The object
\be
J_{\psi_+} := \frac{1}{|V|} \Sigma^i V^i,
\ee
viewed as an endomorphism of $\R^4$, is then a complex structure that corresponds to the pure spinor $\psi_+$. 

The considered case of ${\rm Spin}(4)$ can also be described in terms of quaternions. We will spell out the details in an accompanying paper. 

\subsection{Spin(3,1)}

There is only one possible creation/annihilation operator model for the Lorentz group. It corresponds to the split $\R^{3,1}= \R^2\oplus\R^{1,1}$. We thus introduce one complex coordinate $z$ and one real coordinate $u$, together with the corresponding one-forms $dz, du$. The general spinor is the following polyform
\begin{equation}\label{psi-spin31}
    \Psi=(\alpha+\beta dz\wedge du)+(\gamma du+\delta dz),
\end{equation}
which we write in matrix notations as 
\be
\Psi = \left(\begin{array}{c} \psi_+ \\ \psi_- \end{array}\right), \qquad
\psi_+ = \left(\begin{array}{c} \alpha \\ \beta \end{array}\right), \qquad
\psi_- = \left(\begin{array}{c} \gamma \\ \delta \end{array}\right).
\ee
There are some choices made here, but there are convenient for what follows, as they result in simple expressions for the $\Gamma$-matrices. 

Introducing two pairs of creation/annihilation operators $a, a^\dagger, b, b^\dagger$, the $\Gamma$-operators are as follows
\begin{equation}\label{gamma-31}
    \begin{split}
        \Gamma_1 &= a + a^{\dagger}, \\
        \Gamma_2&= i(a- a^{\dagger}) 
    \end{split}
    \quad,\qquad
    \begin{split}
        \Gamma_3 &= b+ b^{\dagger}, \\ 
        \Gamma_0&= b- b^{\dagger}.
    \end{split}
\end{equation}
Their action on the Dirac spinor is as follows
\be
\Gamma_0 \Psi = \alpha du + \beta dz - \gamma- \delta dz\wedge du, \\ \nonumber
\Gamma_3 \Psi = \alpha du -  \beta dz + \gamma-  \delta dz\wedge du, \\ \nonumber
\Gamma_1 \Psi =  \beta du +\alpha dz+ \delta+ \gamma dz\wedge du, \\ \nonumber
\Gamma_2 \Psi = -\im \beta du + \im \alpha dz - \im \delta+\im \gamma dz\wedge du.
\ee
In matrix notations this becomes
\be
\Gamma_0 =\left( \begin{array}{cc} 0 & -\id \\ \id & 0 \end{array}\right), \quad 
\Gamma_i = \left( \begin{array}{cc} 0 & \sigma^i \\ \sigma^i & 0 \end{array}\right), \quad i=1,2,3.
\ee
This nice form of the resulting $\Gamma$-matrices explains the choices made above. 

The generators of the Lie algebra $\mathfrak{spin}(3,1)$ are the commutators of the above Gamma matrices. The resulting $4\times 4$ matrices are block-diagonal, with the $S_+, S_-$ blocks being respectively
\be
A = - \omega^{0i} \sigma^i + \frac{\im}{2} \omega^{ij} \epsilon^{ijk} \sigma^k, \qquad
A' = \omega^{0i} \sigma^i + \frac{\im}{2} \omega^{ij} \epsilon^{ijk} \sigma^k.
\ee
Both are $2\times 2$ complex tracefree, and $A'=-A^*$. 

The invariant inner product is still given by the formula similar to (\ref{inner-4}), except that the relative sign between the $S_+$ and $S_-$ is reversed, as the result of our usage of $dz\wedge du$ ordering for the top form rather than $du\wedge dz$. Thus, we have
\be
\langle \Psi_1, \Psi_2\rangle = ( \alpha_1 - \beta_1 dz\wedge du + \gamma_1 du + \delta_1 dz) \wedge ( \alpha_2 + \beta_2 dz\wedge du + \gamma_2 du + \delta_2 dz) \Big|_{top} = \\ \nonumber
(\alpha_1 \beta_2-\alpha_2\beta_1) - (\gamma_1 \delta_2 - \gamma_2 \beta_1),
\ee
and so
\be\label{inner-31}
\langle \Psi_1, \Psi_2\rangle = \Psi_1^T \left( \begin{array}{cc} \epsilon & 0 \\ 0 & -\epsilon \end{array}\right) \Psi_2.
\ee

 There are two candidate reality conditions operators $R= \Gamma_0 \Gamma_1 \Gamma_3 \ast$ and $R'= \Gamma_2 \ast$. The second one squares to minus the identity $(R')^2=-\id$, and so is not a suitable reality condition operator. The first one squares to plus the identity $R^2=\id$. It works out to be
\be
R = \left( \begin{array}{cc} 0 & \epsilon \\ -\epsilon & 0 \end{array}\right)\ast,
\ee
where $\epsilon$ is the $2\times 2$ anti-symmetric matrix defined in (\ref{inner-4}). The Majorana spinors are thus of the form
\be
\Psi_M = \left( \begin{array}{c} \psi_+ \\ - \epsilon \psi_+^* \end{array}\right). 
\ee
 
 The canonical pure spinors are $\id$ and $dz\wedge du$, both in $S_+$. The stabiliser of $dz\wedge du$ are lower-diagonal complex $2\times 2$ matrices with the identity on the diagonal, the stabiliser of $\id$ is upper-diagonal matrices. A general Weyl spinor is pure, and ${\rm Spin}(3,1)$ acts transitively on the space of Weyl spinors (of either helicity). 
 
It is interesting to ask what is the geometric information stored in a Weyl spinor, and how to recover it. There is no invariant norm that can be constructed from a Weyl spinor, because the pairing in $S_+$ is anti-symmetric, and there is no invariant operation that can map $S_+$ to itself, as the $R$ operator in this signature maps $S_+$ to $S_-$. But we can construct a vector $\langle R(\psi_+), \Gamma \psi_+\rangle$. A simple calculation gives the following 4-vector
\be\label{vec-31}
V_{\psi_+}:= \langle R(\psi_+), \Gamma \psi_+\rangle=- (\psi_+^\dagger \psi_+, \psi_+^\dagger \sigma^i \psi_+).
\ee
This vector is null, and for $\psi_+=dz\wedge du=(0,1)$ is given by $V_{(0,1)} = - (1,0,0,-1)$, while for $\psi_+=(1,0)$ we have instead $V_{(1,0)} = - (1,0,0,1)$. 

Another geometric object that can be constructed from a Weyl spinor is the two-form $\langle \psi_+, \Gamma\Gamma \psi_+\rangle $. A simple calculation gives
\be
\langle \psi_+, \Gamma\Gamma \psi_+ \rangle= \im \Sigma^i V^i_{\psi_+},
\ee
where
\be
\Sigma^i = \im dt \wedge dx^i + \frac{1}{2} \epsilon^{ijk} dx^j \wedge dx^k
\ee
is the basis of SD two-forms in $\R^{3,1}$ and
\be\label{vec-tilde-31}
V^i_{\psi_+} := \psi_+^T \epsilon \sigma^i \psi_+= (\alpha^2-\beta^2, \im(\alpha^2+\beta^2), - 2\alpha\beta)
\ee
is a complex null 3-vector. By Cartan's theorem \ref{thm:cartan}, we know that $\langle \psi_+, \Gamma\Gamma \psi_+ \rangle=B_2(\psi_+,\psi_+)$ must be decomposable and given by the product of two null directions of $M(\psi_+)$. The real null direction has already been recovered as $V_{\psi_+}$ in (\ref{vec-31}). The other (complex) null direction is essentially given by (\ref{vec-tilde-31}). Indeed, we have
\be
\im \Sigma^i V^i_{\psi_+} = V_{\psi_+} \wedge \tilde{V}_{\psi_+},
\ee
where
\be
\tilde{V}_{\psi_+} = \frac{-1}{\psi_+^\dagger \psi_+} ( 0, V^i_{\psi_+}).
\ee
When $\psi_+=(0,1)$ we have $\tilde{V}_{\psi_+}=(1,-\im,0)$. For the other canonical spinor $\psi_+=(1,0)$ the complex null vector is instead $\tilde{V}_{\psi_+}=(1,\im,0)$. Thus, as expected from general considerations, we can recover both null directions of $M(\psi_+)$ from the generic spinor $\psi_+$, which is pure. 

When the signature of $\R^{r,s}$ is not definite we can only recover from a pure spinor $\psi$ the corresponding null directions of $M(\psi)$. However, the described model with its choice of $\Gamma$-matrices (\ref{gamma-31}) gives not only the null directions $z,u$, but also the complimentary null directions $\bar{z},v$. We have seen how $z,u$ can be recovered from $\psi_+=(0,1)$. We have also seen how the complementary null directions arise from the other canonical pure spinor $\psi_+=(1,0)$. We can then say that a creation/annihilation operator model is in one-to-one correspondence with a choice of two pure spinors $\psi_{1,2}$ (in the case of $\R^{1,3}$ of the same helicity) that satisfy $\langle \psi_1,\psi_2\rangle \not=0$. Each of these pure spinors defines a pair of null vectors, and together the resulting four null vectors span $\R^{1,3}$. 

Given a pair $\psi_{1,2}$ such that $\langle \psi_1,\psi_2\rangle \not=0$ we can construct $\langle \psi_1, \Gamma\Gamma \psi_2\rangle = B_2(\psi_1, \psi_2)$. We have
\be
\langle \psi_1, \Gamma\Gamma \psi_2\rangle=\im \Sigma^i \psi_1^T \epsilon \sigma^i \psi_2
\ee
For $\psi_1=(0,1), \psi_2=(1,0)$ this gives 
\be
-\im \Sigma^3 = dt\wedge dz - \im dx\wedge dy.
\ee
This is a complex 2-form. If we raise one of its indices to convert it to an endomorphism of $\R^{1,3}_\C$, we can write the resulting operator in terms of its real and imaginary parts as $K=I-\im J$. The two operators satisfy $IJ=JI=0$, the real part $I$ squares to the identity times the projector on the $t,z$ plane, and the imaginary part $J$ squares to minus the identity times the projector on the $x,y$ plane
\be
I^2 = P_{t,z}, \qquad J^2=-P_{x,y}.
\ee
Altogether we have $K^2=P_{t,z}+P_{x,y}=\id$. This operator is also metric compatible in the sense 
\be
(Kv,Kv)=-(v,v).
\ee
Thus, it gives an example of what in Section \ref{sec:comp-paracomp} was called a complex/paracomplex structure. Both the real and imaginary parts have (even-dimensional) kernels, and act as normal paracomplex and complex structures on the complements of these kernels. 

The choice of a model is then equivalent to a choice of a complex/paracomplex structure, and this is in turn equivalent to a choice of two pure spinors such that $\langle \psi_1,\psi_2\rangle \not=0$.

\subsection{Spin(2,2)- maximal real index model}

The case of ${\rm Spin}(2,2)$ is particularly interesting because it can be described by two different models. The most familiar model is real, and is based on the choice of a real MTN subspace of $\R^{2,2}$. Let us denote the corresponding null coordinates by $u_1,u_2$. We introduce two pairs of creation/annihilation operators $b_{1,2}, b_{1,2}^\dagger$ and write the $\Gamma$-operators as
\begin{equation}
    \begin{split}
        \Gamma_1&=b_1+ b_1^{\dagger}\\
        \Gamma_3&=b_1 -b_1^{\dagger}
    \end{split}
    \quad,\qquad
    \begin{split}
        \Gamma_2&=b_2+ b_2^{\dagger}\\
        \Gamma_4&=b_2-b_2^{\dagger}.
    \end{split}
\end{equation}
The Dirac spinor is the polyform
\begin{equation}
    \Psi=\alpha + \beta du_{12} +\gamma du_1 + \delta du_2, 
\end{equation}
with $\alpha,\beta,\gamma,\delta\in\R$. In matrix notations the $\Gamma$-matrices are real $4\times 4$, given by
\be
\Gamma_1= \left( \begin{array}{cc} 0 & \id \\ \id & 0 \end{array}\right), \quad \Gamma_2= \left( \begin{array}{cc} 0 & \epsilon \\ -\epsilon & 0 \end{array}\right), \quad \Gamma_3= \left( \begin{array}{cc} 0 & -\sigma^3 \\ \sigma^3 & 0 \end{array}\right), \quad \Gamma_4= \left( \begin{array}{cc} 0 & -\sigma^1 \\ \sigma^1 & 0 \end{array}\right).
\ee
The commutators generate the Lie algebra $\mathfrak{spin}(2,2)$. The diagonal $2\times 2$ blocks are
\be
A= -(\omega^{12} +\omega^{34})\epsilon + (\omega^{13} -\omega^{42})\sigma^3 +(\omega^{14} -\omega^{23})\sigma^1, \\ \nonumber
A' = (\omega^{12} -\omega^{34})\epsilon - (\omega^{13} +\omega^{42})\sigma^3 -(\omega^{14} +\omega^{23})\sigma^1, 
\ee
which exhibits the split  $\mathfrak{spin}(2,2)=\mathfrak{sl}(2)\oplus\mathfrak{sl}(2)$. The inner product is still given by (\ref{inner-4}). There are no non-trivial anti-linear operators, and everything is explicitly real. 

Any Weyl spinor is pure. The canonical pure spinors of this model are $du_{12}$ and $\id$, both in $S_+$. The first of them is stabilised by the copy of ${\rm SL}(2,\R)$ that acts trivially on $S_+$ times the nilpotent subgroup of the other ${\rm SL}(2,\R)$ that consists of lower-diagonal matrices with unity on the diagonal. The action of ${\rm Spin}(2,2)$ is transitive on both $S_+$ and $S_-$. 

Given a Weyl spinor $\psi_+$, we can recover the corresponding null subspace $M(\psi_+)$ by computing the two-form $B_2(\psi_+,\psi_+)$. This form is factorizable, and given by product of the two real null-vectors that generate $M(\psi_+)$. Indeed, we have
\be
 B_2(\psi_+,\psi_+)= (\alpha^2+\beta^2)(dx^1 dx^2+dy^1 dy^2)  -2 \alpha\beta  (dx^1 dy^1 + dx^2 dy^2) + (\alpha^2-\beta^2)(dx^1 dy^2 - dx^2 dy^1),
\ee
where the wedge product of forms is implied, and $x^{1,2}, y^{1,2}$ are coordinates so that
\be
ds^2=(dx^1)^2+(dx^2)^2-(dy^1)^2-(dy^2)^2.
\ee
It is easy to see that $B_2(\psi_+,\psi_+)\wedge B_2(\psi_+,\psi_+)=0$, so it is factorizable. There is ambiguity as to the choice of the null directions giving this two-form. One possible choice is to take one of the two null vectors to not involve the $dy^2$ component. Then this null vector is a multiple of
\be
V=  - 2\alpha\beta dx^2 + (\alpha^2-\beta^2) dx^1+ (\alpha^2+\beta^2) dy^1,
\ee
which is indeed null. Then 
\be
B_2(\psi_+,\psi_+)= V\wedge \tilde{V},
\ee
where
\be
\tilde{V} = \frac{1}{\alpha^2+\beta^2} (  2\alpha\beta dx^1 + (\alpha^2-\beta^2) dx^2+ (\alpha^2+\beta^2) dy^2) .
\ee

A single pure spinor $\psi$ only allows to recover its null subspace $M(\psi)$. A model that we started from comes with two MTN complementary subspaces. One of them corresponds to the canonical spinor $du_{12}$, the other to the canonical spinor $\id$. More generally, a pair of two pure spinors is in correspondence with a paracomplex structure whose eigenspaces are two complementary MTN. In the other direction, a pair of two pure spinors $\psi_1, \psi_2$ defines a paracomplex structure. For example, explicitly, taking $\psi_1=du_{12}, \psi_2=\id$ we compute 
\be
B_2(\psi_1,\psi_2) =- dx^1 dy^1- dx^2 dy^2.
\ee
Raising one of the indices we get an operator that squares to plus the identity and is metric-compatible in the sense of (\ref{PC-compat}). So, as expected, a single MTN subspace is recoverable from a single pure spinor, while a paracomplex structure giving rise to a complementary pair of MTN subspaces is recoverable from a pair $\psi_1, \psi_2$ satisfying $\langle\psi_1,\psi_2\rangle \not=0$. The considered creation/annihilation operator model is based on a choice of a paracomplex structure, and comes with a preferred pair of two pure spinors. 

\subsection{Spin(2,2) - index zero model}

To construct this model we choose an MTN subspace spanned by two complex vectors, obtained as $-\im$ eigenvalue eigenvectors of a complex structure on $\R^{2,2}$. We refer to the corresponding null complex coordinates $z_{1,2}$, and the corresponding one-forms $dz_{1,2}$. The general Dirac spinor is the polyform
\begin{equation}
    \Psi=\alpha + \beta dz_{12} +\gamma dz_1 + \delta dz_2, 
\end{equation}
with $\alpha,\beta,\gamma,\delta\in\C$. We introduce two pairs of creation/annihilation operators $a_{1,2}, a^\dagger_{1,2}$. The $\Gamma$-operators are given by
\begin{equation}
    \begin{split}
        \Gamma_1&=a_1+ a_1^{\dagger}\\
        \Gamma_3&=a_2-a_2^{\dagger}
    \end{split}
    \quad,\qquad
    \begin{split}
        \Gamma_2&=-\im( a_1- a_1^{\dagger}) \\
        \Gamma_4&=\im(a_2+a_2^{\dagger}).
    \end{split}
\end{equation}
Here our choice of $\Gamma_{1,2}$ is motivated to match $\Gamma_{4,3}$ in (\ref{gamma-spin4}). The other choices are motivated by the desire to have nicer looking $\Gamma$-matrices. The $\Gamma$-matrices are then easily recoverable from (\ref{gamma-matr-spin4}) and are given by
\be
\Gamma_1= \left( \begin{array}{cc} 0 & \id \\ \id & 0 \end{array}\right), \quad \Gamma_2= \left( \begin{array}{cc} 0 & \im\sigma^3 \\ -\im\sigma^3 & 0 \end{array}\right), \quad \Gamma_3= \left( \begin{array}{cc} 0 & -\sigma^1 \\ \sigma^1 & 0 \end{array}\right), \quad \Gamma_4= \left( \begin{array}{cc} 0 & -\sigma^2 \\ \sigma^2 & 0 \end{array}\right).
\ee
The $2\times 2$ blocks of the Lie algebra element are then
\be
A= -\im \sigma^3( \omega^{12}+\omega^{34}) + \sigma^1 (\omega^{13}-\omega^{42}) + \sigma^2(\omega^{14}-\omega^{23}), \\ \nonumber
A'= \im \sigma^3( \omega^{12}-\omega^{34}) - \sigma^1 (\omega^{13}+\omega^{42}) - \sigma^2(\omega^{14}+\omega^{23}).
\ee
Both are tracefree matrices with imaginary diagonal and the off-diagonal elements being complex conjugates of each other. These matrices form $\mathfrak{su}(1,1)$, and so the Lie algebra is $\mathfrak{spin}(2,2)=\mathfrak{su}(1,1)\oplus\mathfrak{su}(1,1)$. Thus, this version of the creation/annihilation operator model exhibits the isomorphism ${\rm Spin}(2,2)={\rm SU}(1,1)\times{\rm SU}(1,1)$. The invariant inner product is still given by (\ref{inner-4}).

The novelty as compared to the previous real index two model is that the spinors are now complex. However, there are now two non-trivial anti-linear operators that can be constructed $R= \Gamma_1\Gamma_3 \ast$ and $R'=\Gamma_2\Gamma_4 \ast$. Unlike the case of ${\rm Cl}(4)$ where their analogs both square to minus the identity, now they both square to plus the identity, and either one of them can be used to define the notation of Majorana spinors. For concreteness, let us use $R'$ as the reality condition operator. In matrix form we have
\be
R' = \left(\begin{array}{cc} \sigma^1 & 0 \\ 0 & \sigma^1 \end{array} \right) \ast.
\ee
The operator $R$ has the same action on $S_+$, and is minus this on $S_-$. The action of $R'$ preserves $S_+$ (and $S_-$), and allows us to define Majorana-Weyl spinors. It is clear that a Majorana-Weyl spinor in both $S_+, S_-$ is of the form
\be\label{psi-MW-22}
\psi_{MW} = \left( \begin{array}{c} \alpha \\ \alpha^* \end{array}\right). 
\ee
Thus, a Majorana-Weyl spinor is parametrised by a single complex number. This shows that Majorana spinors in the case of this model are not different from the case of the index two model, where they are parametrised by two real numbers. 

However, a general spinor in the case of this model is complex 2-dimensional, and we need such complex spinors to recover the complex structure from a pure spinor. Thus, we take a general Weyl spinor $\psi_+\in S_+$. Having the anti-linear operators $R,R'$ in our disposal (which agree on $S_+$) we can compute
\be\label{norm-22}
\langle R(\psi_+), \psi_+\rangle = |\beta|^2-|\alpha|^2.
\ee
The action of ${\rm Spin}(2,2)$ on $S_+$ viewed as complex 2-component columns preserves this invariant. The orbit corresponding to a fixed value of this invariant is a  $|\beta|^2-|\alpha|^2=const$ is $AdS_3$ as a manifold. One of the two copies of ${\rm SU}(1,1)$ does not act on $S_+$, while the other acts transitively on $AdS_3$. Thus, the stabiliser of any point on the orbit $|\beta|^2-|\alpha|^2=const$ is ${\rm SU}(1,1)$ and 
\be
AdS_3= {\rm Spin}(2,2)/ {\rm SU}(1,1)={\rm SU}(1,1).
\ee
This completely analogous to what we had in the $\R^4$ case where the analogous statement was $S^3={\rm SU}(2)$. What is worth stressing is that all this becomes possible only in the setting of general Weyl spinors without any Majorana reality condition imposed. 

Other geometric date stored in a general Weyl spinor are as follows. We compute 
\be
B_2(\psi_+,\psi_+)=2\im \alpha \beta( dx^1 dx^2 + dy^1 dy^2) + (\alpha^2-\beta^2) ( dx^1 dy^1 + dx^2 dy^2) + \im(\alpha^2+\beta^2) (dx^1 dy^2- dx^2 dy^1).
\ee
This two-form is decomposable. For example, for the canonical spinor $\psi_+=(0,1)$ we get $B_2(\psi_+,\psi_+)= -(dx^1+\im dx^2)(dy^1-\im dy^2)$, and for $\psi_+=(1,0)$ we have $B_2(\psi_+,\psi_+)= (dx^1-\im dx^2)(dy^1+\im dy^2)$. Finally, the other meaningful object we can construct from a general Weyl spinor is 
\be\nonumber
B_2(R(\psi_+),\psi_+)=\im (|\alpha|^2+ |\beta|^2) ( dx^1 dx^2 + dy^1 dy^2) -2\im {\rm Im}(\alpha^*\beta)  ( dx^1 dy^1 + dx^2 dy^2) + 2\im {\rm Re}(\alpha^*\beta) (dx^1 dy^2- dx^2 dy^1).
\ee
This is a pure imaginary 2-form, which can be interpreted as a complex structure when one of its indices is raised and it is rescaled appropriately. The best way to do this is to introduce a triple of self-dual two-forms
\be
\Sigma^3 = dx^1 dx^2 + dy^1 dy^2, \quad  \Sigma^1= dx^1 dy^2- dx^2 dy^1, \quad \Sigma^2 = dx^1 dy^1 + dx^2 dy^2, 
\ee
Then
\be
B_2(R(\psi_+),\psi_+)=\im \Sigma^i V_i, 
\ee
where
\be
V_i = ( 2{\rm Re}(\alpha^*\beta) , - 2{\rm Im}(\alpha^*\beta) , |\alpha|^2+ |\beta|^2).
\ee
We note that when $\psi_+$ is a spinor of fixed norm (\ref{norm-22}) and thus a point on $AdS_3$ of a fixed radius of curvature, the vector $V_i$ is a point on the upper sheet of the two-sheeted hyperboloid
\be
(V_1)^2+(V_2)^2 - (V_3)^2 = - (|\alpha|^2- |\beta|^2)^2,
\ee
and thus a point on the hyperbolic plane $H_2$. We thus encounter an instance of the non-compact version of the Hopf fibration
\be
S^1 \to AdS_3 \to H_2,
\ee
which is a precise analog of the usual $S^1\to S^3\to S^2$ that was encountered in the case of $\R^4$.  

The objects $\Sigma^i$, viewed as endomorphisms of $\R^{2,2}$ satisfy
\be
\Sigma^3_\mu{}^\rho \Sigma^3_\rho{}^\nu = - \delta_\mu{}^\nu, \quad
\Sigma^1_\mu{}^\rho \Sigma^1_\rho{}^\nu = \delta_\mu{}^\nu, \quad
\Sigma^2_\mu{}^\rho \Sigma^2_\rho{}^\nu =  \delta_\mu{}^\nu.
\ee
This shows that
\be
J_\mu{}^\nu : = \frac{1}{|\alpha|^2-|\beta|^2} \Sigma^i_\mu{}^\nu V_i
\ee
squares to minus the identity and is a complex structure. In particular, when $\psi_+=(1,0)$ or $\psi_+=(0,1)$ this complex structure is plus or minus $\Sigma^3$. Thus, complex structures on $\R^{2,2}$ are parametrised by points on the hyperbolic plane $H_2$, as in the case of $\R^4$ they are parametrised by points of $S^2$. Once again, it needs to be emphasised that we have access to this complex picture only when we consider general complex-valued Weyl spinors.

It is interesting to see how much of the above picture survives if we impose the Majorana condition. First, considering the space of Majorana-Weyl spinors of the form (\ref{psi-MW-22}), the action of ${\rm SU}(1,1)$ on this space is transitive, with real dimension one stabiliser. For example, the stabiliser of the spinor $\psi_{MW}=(1,1)$ is the subgroup of matrices of the form
\be
\left( \begin{array}{cc} 1+\im \xi & -\im \xi \\ \im\xi & 1-\im \xi \end{array}\right), \quad \xi\in \R.
\ee
The full stabiliser of a Majorana-Weyl spinor in ${\rm Spin}(2,2)$ is then ${\rm SU}(1,1)\times\R$. Geometrically, if a general Weyl spinor (of fixed norm) represents a point on $AdS_3$, Majorana-Weyl spinors have zero norm and correspond to points on the light-cone of a point in $AdS_3$.

To see geometric data encoded by a Majorana-Weyl spinor we take a general Majorana-Weyl spinor (\ref{psi-MW-22}) in $S_+$, and compute the two-form $B_2(\psi_{MW},\psi_{MW})$. We have
\be\nonumber
B_2(\psi_{MW},\psi_{MW})=2\im |\alpha|^2( dx^1 dx^2 + dy^1 dy^2) + (\alpha^2-(\alpha^*)^2) ( dx^1 dy^1 + dx^2 dy^2) + \im(\alpha^2+(\alpha^*)^2) (dx^1 dy^2- dx^2 dy^1).
\ee
This two-form is purely imaginary and decomposable. For example, for $\alpha=1$ we have
$$B_2(\psi_{MW},\psi_{MW}) = 2\im ( dx^1+dy^1)( dx^2+dy^2).$$
Thus, a Majorana-Weyl spinor only carries information about two real null directions, which are the directions spanning $M(\psi_{MW})$. 

To conclude, we have seen that the index zero model of ${\rm Spin}(2,2)$ tells us that it is in general too restrictive to impose the Majorana condition, even though this is possible. The general complex Weyl spinors are necessary to recover the complex structure on $\R^{2,2}$ that such a model is based on. When one imposes the Majorana condition, the geometry of the model becomes the same as that of the explicitly real index zero model. This tells us that Weyl spinors whose MTN has real index zero are also possible to describe in the index two model, but one must consider complex-valued spinors. Thus, the real index zero model seems to be preferred in the sense that complex-valued spinors are completely natural in it. 

The two types of Weyl spinors that arise in the case of ${\rm Spin}(2,2)$ are only visible when the spinors are complex-valued. The spinors of non-zero norm (\ref{norm-22}) are points in $AdS_3$ of a fixed radius of curvature, and $AdS_3={\rm SU}(1,1)$. The spinors of zero norm (\ref{norm-22}) are Majorana-Weyl spinors and ${\rm SU}(1,1)$ acts on this orbit with a non-trivial stabiliser. Spinors of the first type are in correspondence with MTN subspaces of $\R^{2,2}$ of real index zero. Spinors of the second type are in correspondence with MTN subspaces of real index two. So, the lesson is that we lose some interesting geometry if we impose the Majorana condition. 

\section{Dimension six}
\label{sec:six}

\subsection{Spin(6)}

As usual, only real index zero model is possible in this case. We choose a complex structure on $\R^6$, and introduce 3 complex null coordinates $z_{1,2,3}$, as well as the corresponding one-forms $dz_{1,2,3}$. We introduce 3 pairs of creation/annihilation operators $a_{1,2,3}, a_{1,2,3}^\dagger$. The Dirac spinor is a polyform 
\begin{equation}\label{psi-6}
    \Psi= \alpha_1 dz_{23} + \alpha_2 dz_{31}+\alpha_3 dz_{12}+\alpha_4+\beta_1 dz_1 +\beta_2 dz_2 + \beta_3 dz_3 - \beta_4 dz_{123},
\end{equation}
with all coefficients complex-valued. The reason why the last term is included with the minus sign will become clear when we consider the inner product. 
We act upon this Dirac spinor with the following Gamma matrices
\begin{equation}\label{Gamma-6}
    \begin{split}
        \Gamma_1 &= a_1+a_1^{\dagger}, \\
        \Gamma_4&= i(a_1-a_1^{\dagger}) 
    \end{split}
    \quad,\qquad
    \begin{split}
        \Gamma_2 &= a_2+a_2^{\dagger}, \\ 
        \Gamma_5&= i(a_2-a_2^{\dagger})
    \end{split}
    \quad,\qquad
    \begin{split}
        \Gamma_3&= a_3+a_3^{\dagger}, \\ 
        \Gamma_6&= i(a_3-a_3^{\dagger})
    \end{split}
\end{equation}
All $\Gamma$-matrices work out to be 
\be\label{Gamma-6*}
\Gamma_I = \left( \begin{array}{cc} 0 & \gamma_I \\ \gamma_I^\dagger & 0 \end{array}\right),\qquad I=1,\ldots, 6,
\ee
where $\gamma_I$ are the following $4\times 4$ matrices
\be\label{gamma-6}
\gamma_1 = \left( \begin{array}{cccc} 0 & 0 & 0 & -1 \\ 0 & 0 & -1 & 0\\ 0 & 1 & 0 & 0 \\ 1 & 0 & 0 & 0\end{array}\right), \quad
\gamma_2 = \left( \begin{array}{cccc} 0 & 0 & 1 & 0 \\ 0 & 0 & 0 & -1\\ -1 & 0 & 0 & 0 \\ 0 & 1 & 0 & 0\end{array}\right), \quad
\gamma_3 = \left( \begin{array}{cccc} 0 & -1 & 0 & 0 \\ 1 & 0 & 0 & 0\\ 0 & 0 & 0 & -1 \\ 0 & 0 & 1 & 0\end{array}\right), \\ \nonumber
\gamma_4 = \im\left( \begin{array}{cccc} 0 & 0 & 0 & 1 \\ 0 & 0 & -1 & 0\\ 0 & 1 & 0 & 0 \\ -1 & 0 & 0 & 0\end{array}\right), \quad
\gamma_5 = \im\left( \begin{array}{cccc} 0 & 0 & 1 & 0 \\ 0 & 0 & 0 & 1\\ -1 & 0 & 0 & 0 \\ 0 & -1 & 0 & 0\end{array}\right), \quad
\gamma_6 = \im\left( \begin{array}{cccc} 0 & -1 & 0 & 0 \\ 1 & 0 & 0 & 0\\ 0 & 0 & 0 & 1 \\ 0 & 0 & -1 & 0\end{array}\right).
\ee
They are all anti-symmetric. The commutator of these $\Gamma$-matrices is block-diagonal, with anti-hermitian tracefree $4\times 4$ blocks on the diagonal. This exhibits the isomorphism $\mathfrak{spin}(6)=\mathfrak{su}(4)$. 

The inner product pairs even to odd polyforms, and so is a pairing $\langle S_+,S_-\rangle$. Explicitly, we get
\be\label{inner-6}
\langle \Psi, \tilde{\Psi}\rangle =- \sum_{I=1}^4 \alpha_I \tilde{\beta}_I + \sum_{I=1}^4 \beta_I \tilde{\alpha}_I.
\ee
It is in order to have the same signs here that we have put the minus sign in the last term in (\ref{psi-6}). 

There are two anti-linear operators that can be constructed $R=\Gamma_1\Gamma_2\Gamma_3\ast$ and $R'=\Gamma_4\Gamma_5\Gamma_6\ast$. The first of these squares to minus the identity, while $(R')^2=\id$. So, it is $R'$ that gives us a good real structure. It works out to be given by
\be
R' = \left(\begin{array}{cc} 0 & \id \\ \id & 0 \end{array}\right) \ast.
\ee

Given a Weyl spinor $\psi_+\in S_+$, we can construct
\be\label{norm-6}
\langle R'(\psi_+), \psi_+\rangle = \sum_{I=1}^4 |\alpha_I|^2.
\ee
Thus, there is a positive-definite Hermitian invariant quadratic form on $S_+$. The group ${\rm Spin}(6)={\rm SU}(4)$ acts on the subset in $S_+$ of spinors of fixed norm squared transitively, with the stabiliser ${\rm SU}(3)$. Thus, we have
\be
S^7 = {\rm SU}(4)/{\rm SU}(3).
\ee
Given that Weyl spinors are pure in this dimension, and directions of pure spinors define complex structures in $\R^6$, we see that the space of complex structures on $\R^6$ is $S^7$. 

Given a Weyl spinor, the object $B_1(\psi_+,\psi_+)$ vanishes because the matrices (\ref{gamma-6}) are anti-symmetric. The only non-vanishing object that can be constructed without using the operator $R'$ is $B_3(\psi_+,\psi_+)$. This means that Weyl spinors are pure in this dimension. From general grounds we know that $B_3(\psi_+,\psi_+)$ is given by the wedge product of the three complex null directions in $M(\psi_+)$. The objects that can be constructed using $R'$ are the norm (\ref{norm-6}) as well as $B_2(R'(\psi_+),\psi_+)$. This is a two-form that gives the complex structure that corresponds to $\psi_+$ when one of its indices is raised and it is rescaled appropriately. For example, for the canonical spinor $\id\in S_+$ one gets
\be
B_2(R'(\id),\id)= \im (dx^1 \wedge dx^4+ dx^2\wedge dx^5+dx^3 \wedge dx^6), 
\ee
while for $dz_{123}\in S_-$ the result is the same with the extra minus sign in front. 

\subsection{Spin(5,1)}

There is only a single possible model arising in this case, this is the model of real index one. We thus need two complex and one real null directions. We denote the corresponding complex coordinates by $z_{1,2}$ and the real coordinate by $u$. The general polyform is given by
\begin{equation}\label{psi-6}
    \Psi= \alpha_1 dz_{2}\wedge du + \alpha_2 du\wedge dz_{1}+\alpha_3 dz_{12}+\alpha_4+\beta_1 dz_1 +\beta_2 dz_2 + \beta_3 du - \beta_4 dz_{12}\wedge du,
\end{equation}
where we kept the same signs as in (\ref{psi-6}) in order for the inner product (\ref{inner-6}) to be unchanged. All coefficients are still complex. We introduce two pairs $a_{1,2}, a_{1,2}^\dagger$ of creation/annihilation operators for $dz_{1,2}$ and one pair $b,b^\dagger$ for $du$. The $\Gamma$-operators are given by
\begin{equation}\label{Gamma-51}
    \begin{split}
        \Gamma_1 &= a_1+a_1^{\dagger}, \\
        \Gamma_4&= i(a_1-a_1^{\dagger}) 
    \end{split}
    \quad,\qquad
    \begin{split}
        \Gamma_2 &= a_2+a_2^{\dagger}, \\ 
        \Gamma_5&= i(a_2-a_2^{\dagger})
    \end{split}
    \quad,\qquad
    \begin{split}
        \Gamma_3&= b+b^{\dagger}, \\ 
        \Gamma_6&= b-b^{\dagger}.
    \end{split}
\end{equation}
The only modification as compared to (\ref{Gamma-6}) is in the $\Gamma_6$ operator. So, $\Gamma_1,\ldots,\Gamma_5$ continue to be given by (\ref{Gamma-6*}), while 
\be
\Gamma_6 = \left( \begin{array}{cc} 0 & \tilde{\gamma}_6 \\ \tilde{\gamma}_6 & 0 \end{array}\right), \qquad 
\tilde{\gamma}_6=-\im \gamma_6,
\ee
where $\gamma_6$ is given by (\ref{gamma-6}). The arising Lie algebra matrices do not have a particularly nice characterisation unless we bring in quaternions, but this is the subject of the accompanying paper. So, we refrain from spelling out the Lie algebra matrices explicitly. 

The possible anti-linear operators are $R=\Gamma_1\Gamma_2\Gamma_3\Gamma_6 \ast$ and $R'=\Gamma_4\Gamma_5 \ast$. Both of them now squares to minus the identity, so there are no Majorana spinors in this signature. For later purposes, let us record that $R'$ as an operator on $S_+$ acts as
\be
R' \Big|_{S_+} = \left( \begin{array}{cc} \epsilon & 0 \\ 0 & -\epsilon \end{array}\right) \ast.
\ee
Both $R,R'$ operators preserve the helicity. This means that there is no quadratic invariant that can be constructed for a Weyl spinor in this signature. There are no higher degree invariants either, and ${\rm Spin}(5,1)$ acts on its Weyl spinors transitively. In fact, using quaternions one can identify ${\rm Spin}(5,1)={\rm SL}(2,\Hq)$, and $S_\pm=\Hq^2$. 

Because $\gamma_1,\ldots,\gamma_5,\tilde{\gamma}_6$ are all anti-symmetric $B_1(\psi_+,\psi_+)=0$, and the only non-trivial object that can be constructed without using the $R,R'$ operators is $B_3(\psi_+,\psi_+)$. So, Weyl spinors are pure in this dimension. The three-form $B_3(\psi_+,\psi_+)$ is given by the product of two complex and one real null directions in $M(\psi_+)$. 

The real null direction can be recovered by computing $B_1(R(\psi_+), \psi_+)$. It is an instructive computation, so we spell it out. A computation gives the following vector 
\be
B_1(R(\psi_+), \psi_+) = 2dx^1  {\rm Re}(\alpha_3 \alpha_1^* - a_4 \alpha_2^* ) + 2dx^2 {\rm Re}(\alpha_1 \alpha_4^* + a_3 \alpha_2^* ) + dx^3 (-|\alpha_1|^2-|\alpha_2|^2+|\alpha_3|^2+|\alpha_4|^2) 
\\ \nonumber
+ 2dx^4 {\rm Im}(\alpha_3 \alpha_1^* - a_4 \alpha_2^* )+2dx^5 {\rm Im}(\alpha_1 \alpha_4^* + a_3 \alpha_2^* ) + dy (|\alpha_1|^2+|\alpha_2|^2+|\alpha_3|^2+|\alpha_4|^2),
\ee
which is a null vector. The $\R^5$ part of this vector, for a spinor satisfying $|\alpha_1|^2+|\alpha_2|^2+|\alpha_3|^2+|\alpha_4|^2=1$ lies on $S^4\subset \R^5$. This gives a projection $S^7\to S^4$, which is the quaternionic Hopf fibration. This becomes very clear if one uses the quaternionic formalism to describe spinors in this signature.

The only geometric information a Weyl spinor encodes is that about its totally null subspace $M(\psi_+)$. When the metric is not definite, this is not enough to recover the complementary MTN. But given two Weyl spinors $\psi_+\in S_+, \psi_-\in S_-$ satisfying $\langle \psi_+,\psi_-\rangle\not =0$ we can recover both $M(\psi_+), M(\psi_-)$ as the eigenspaces of a structure of a mixed type. This is a complex-valued operator that can be recovered from $B_2(\psi_+,\psi_-)$. Indeed, taking the two canonical pure spinors of our model $\psi_+=\id, \psi_- = dz_{12} du$ we have
\be
B_2(\id, dz_{12} du) = \im (dx^1\wedge dx^3+dx^2\wedge dx^4) + dx^3 \wedge dy.
\ee
This is a sum of complex structure (multiplied by the imaginary unit) on $\R^4$ spanned by $x^{1,2,3,4}$, and a paracomplex structure on $\R^{1,1}$ spanned by $x^3, y$. Its eigenspaces recover both $M(\id), M(dz_{12} du)$. Again, this gives an example of a structure of the mixed type as we described in Section \ref{sec:comp-paracomp}.

\subsection{Spin(4,2)}

There are two possible models in this case. In the real index two model one takes one complex and two real directions as those spanning the MTN. In the real index zero modes there are no real directions. In both cases all the coefficients of the polyform are complex-valued. It should make no difference which model is used, except that the two different types of MTN's possible in this signature will be realised differently in each model. We will spell out the details of the real index zero model, with 3 complex directions. 

There are thus three complex coordinates $z_{1,2,3}$ as in the $\R^6$ case. The polyform representing the Dirac spinor is still given by (\ref{psi-6}), and the inner product is still (\ref{inner-6}). The $\Gamma$-operators are now 
\begin{equation}\label{Gamma-42}
    \begin{split}
        \Gamma_1 &= a_1+a_1^{\dagger}, \\
        \Gamma_4&= \im(a_1-a_1^{\dagger}) 
    \end{split}
    \quad,\qquad
    \begin{split}
        \Gamma_2 &= a_2+a_2^{\dagger}, \\ 
        \Gamma_5&= \im(a_2-a_2^{\dagger})
    \end{split}
    \quad,\qquad
    \begin{split}
        \Gamma_3&= \im(a_3+a_3^{\dagger}), \\ 
        \Gamma_6&= a_3-a_3^{\dagger}.
    \end{split}
\end{equation}
We have numbered the $\Gamma$-operators so that they match those in (\ref{Gamma-51}), apart from 
\be
\Gamma_3 =\im \left( \begin{array}{cc} 0 & \gamma_3 \\ \gamma_3^\dagger & 0 \end{array}\right), \qquad
\Gamma_6 = -\im \left( \begin{array}{cc} 0 & \gamma_6 \\ \gamma_6^\dagger & 0 \end{array}\right),
\ee
where $\gamma_3, \gamma_6$ are still given by (\ref{gamma-6}). Thus, the directions $3,6$ are the negative definite directions in $\R^{4,2}$. 

The inner product is still given by (\ref{inner-6}). To understand what the Lie algebra works out to be, it is helpful to start by discussing the possible reality conditions first. The anti-linear operators are $R= \Gamma_1\Gamma_2 \Gamma_6 \ast, R'=\Gamma_3\Gamma_4\Gamma_5\ast$. They both square to the identity, so any one of them can be used to define the notion of Majorana spinors. We have
\be
R= \left( \begin{array}{cc} 0 & \rho \\ \rho & 0 \end{array}\right) \ast, \qquad \rho = \left( \begin{array}{cc} \id & 0 \\ 0 & -\id \end{array}\right) .
\ee
This operator of complex conjugation allows us to define the invariant of a Weyl spinor. Indeed, we have
\be\label{inv-42}
\langle R(\psi_+), \psi_+\rangle = |\alpha_1|^2+|\alpha_2|^2-|\alpha_3|^2-|\alpha_4|^2.
\ee
Given that there is a quadratic invariant, the spin group in this signature becomes ${\rm Spin}(4,2)={\rm SU}(2,2)$. This can also be seen at the level of the Lie algebra by computing the commutators of $\Gamma$-matrices, and writing the general Lie algebra element. 

The only geometrical objects that can be constructed from $\psi_+$ without using $R$ are $B_1(\psi_+,\psi_+)$ and $B_3(\psi_+,\psi_+)$, but the first of these is zero because all $\gamma_I$ are anti-symmetric. The three-form $B_3(\psi_+,\psi_+)$ is given by the product of 3 null directions of spanning $M(\psi_+)$. The objects that can be constructed using the complex conjugation map are the invariant (\ref{inv-42}), as well as the two-form $B_2(R(\psi_+),\psi_+)$. The two-form can be used to recover the complex structure underlying the model, or more generally the complex structure that corresponds to $\psi_+$, when it exits. To see this, we first compute this two-form for the canonical spinor $\psi_+=\id$. We get
\be
B_2(R(\id), \id) = -\im (dx^1\wedge dx^5+dx^2\wedge dx^4+ dx^3 \wedge dx^6) .
\ee
We remind that the directions $x^3, x^6$ are the negative-definite directions. Raising an index (and multiplying by the imaginary unit) we get the complex structure that went into the construction of the model. 

The norm of the spinor $\psi_+=\id$ is non-zero. To see the other possible type of the spinor orbit that would correspond to an MTN of real index two, we consider a Weyl spinor of zero norm. For example, taking $\psi_+ = dz_{23}+ \id$ we get a decomposable 2-form
\be
B_2(R(dz_{23}+ \id), dz_{23}+\id) = 2\im (dx^2+ dx^6)\wedge (dx^3-dx^5).
\ee
Thus, this tells us that a null Weyl spinor gives rise to two real null directions, and thus the corresponding $M(\psi_+)$ is two real and one complex direction. To recover the structure of the mixed type that would give both $M(\psi_+)$ as well as its complement one needs another Weyl spinor $\psi_-$, with $\langle\psi_+, \psi_-\rangle\not=0$. Then $B_2(\psi_+,\psi_-)$ will produce the desired mixed structure. For example, a negative null Weyl spinor that has $\langle \psi_-, dz_{23}+ \id\rangle\not=0$ is $\psi_-= dz_1- dz_{123}$. Then
\be
B_2(dz_1- dz_{123}, dz_{23}+ \id) = 2\im dx^1\wedge dx^4 + 2 dx^2 \wedge dx^6 + 2 dx^3 \wedge dx^5,
\ee
which gives the mixed structure of the type we described in Section \ref{sec:comp-paracomp}, with complex null directions $dx^1\pm \im dx^2$ and real null directions $dx^2\pm dx^6, dx^3\pm dx^5$. 

To summarise, there are two types of pure spinors in this case. Non-null spinors correspond to complex structures on $\R^{4,2}$, while null spinors have MTN consisting of one complex and two real directions. A complex/paracomplex structure with one complex and two real null directions can be recovered from a pair of null spinors satisfying $\langle\psi_+, \psi_-\rangle\not=0$. Both types of spinors can be seen in the same model that works with only complex null directions, so the maximal real index model that would work with two real directions is possible, but not necessary to consider.

\subsection{Spin(3,3)}

As in the previous case, there are two possible models. A model with real index three, that works with 3 real null coordinates, and has all $\Gamma$-operators built from the creation/annihilation operators with real coefficients. We have described this model, in more generality, in Section \ref{sec:n-n}. This model is explicitly real, and it is natural to take in it all spinors to be Majorana(-Weyl). The other model has real index one, and takes two complex and one real null direction. Let us spell out the details of this model.

To make the required modification from the case $\R^{4,2}$ minimal we will call the complex coordinates $z_{1,3}$, and the real one $u$. There are two pairs of creation/annihilation operators $a_{1,3}, a^\dagger_{1,3}$ and one pair $b,b^\dagger$. The general Dirac spinor is obtained from (\ref{psi-6}) by replacing the second coordinate $z_2\to u$
\begin{equation}\label{psi-33}
    \Psi= \alpha_1 du dz_{3} + \alpha_2 dz_{31}+\alpha_3 dz_{1} du+\alpha_4+\beta_1 dz_1 +\beta_2 du + \beta_3 dz_3 - \beta_4 dz_{31} du.
\end{equation}
All coefficients are complex-valued. 
The inner product is then still given by (\ref{inner-6}). We take the $\Gamma$-operators to be
\begin{equation}\label{Gamma-33}
    \begin{split}
        \Gamma_1 &= a_1+a_1^{\dagger}, \\
        \Gamma_4&= \im(a_1-a_1^{\dagger}) 
    \end{split}
    \quad,\qquad
    \begin{split}
        \Gamma_2 &= b+b^{\dagger}, \\ 
        \Gamma_5&= b-b^{\dagger}
    \end{split}
    \quad,\qquad
    \begin{split}
        \Gamma_3&= \im(a_3+a_3^{\dagger}), \\ 
        \Gamma_6&= a_3-a_3^{\dagger}.
    \end{split}
\end{equation}
We thus have only modified the operators $\Gamma_{2,5}$ as compared to the case (\ref{Gamma-42}). The directions $1,2,4$ are now positive-definite, while $3,5,6$ are negative-definite. The modified $\Gamma$-matrices are
\be
\Gamma_2 =\left( \begin{array}{cc} 0 & \gamma_2 \\ \gamma_2^\dagger & 0 \end{array}\right), \qquad
\Gamma_5 = -\im \left( \begin{array}{cc} 0 & \gamma_5 \\ \gamma_5^\dagger & 0 \end{array}\right),
\ee
where $\gamma_2, \gamma_5$ are still given by (\ref{gamma-6}). 

The complex conjugation operators are given by $R=\Gamma_1\Gamma_2\Gamma_5\Gamma_6\ast$ and $R'=\Gamma_3\Gamma_4\ast$. They both square to plus the identity, and so either can be used to define Majorana-Weyl spinors. We have
\be
R= \left(\begin{array}{cc} \rho & 0 \\ 0 & \rho \end{array}\right)\ast, \qquad \rho = \left( \begin{array}{cc} 0 & -\id \\ -\id  & 0 \end{array}\right) .
\ee
The Lie algebra $\mathfrak{spin}(3,3)$ commutes with the real structure defined by $R$, and is isomorphic to $\mathfrak{sl}(4)$. 

As in all cases considered before, the only object that can be constructed from a Weyl spinor $\psi_+\in S_+$ without involving complex conjugation is $B_3(\psi_+,\psi_+)$. It is given by the product of the three null directions spanning $M(\psi_+)$. There are no invariant of $\psi_+$ that can be constructed in this signature, for such an invariant would involve a pairing of $\psi_+$ with some spinor in $S_-$ that must be constructed from $\psi_+$, but there is no such spinor in this signature as the complex conjugation operator $R:S_+\to S_+$. But the action of ${\rm Spin}(3,3)$ on $S_+$, prior to imposing the Majorana-Weyl condition, cannot be transitive. Indeed, we expect two different types of orbits corresponding to two different types of MTN that are possible in this signature. 

To see how this arises, let us compute $B_1(R(\psi_+),\psi_+)$. We get
\be\nonumber
 B_1(R(\psi_+),\psi_+)= (-2{\rm Re}(\alpha_1 \alpha_2^*) + 2{\rm Re}(\alpha_3 \alpha_4^*)) dx^1 
 + (|\alpha_1|^2-|\alpha_2|^2-|\alpha_3|^2+|\alpha_4|^2) dx^2 + ( 2{\rm Im}(\alpha_1\alpha_4^*) - 2{\rm Im}(\alpha_2\alpha_3^*)) dx^3 
 \\ \nonumber
 + ( 2{\rm Im}(\alpha_1\alpha_2^*) +2{\rm Im}(\alpha_3\alpha_4^*)) dx^4 
 + (-|\alpha_1|^2-|\alpha_2|^2+|\alpha_3|^2+|\alpha_4|^2) dx^5 + (2{\rm Re}(\alpha_1 \alpha_4^*) - 2{\rm Re}(\alpha_2 \alpha_3^*)) dx^6 .
 \ee
 This is a null vector in $\R^{3,3}$, which vanishes when the spinor is Majorana-Weyl $\alpha_3= -\alpha_1^*, \alpha_4=-\alpha_2^*$. We also note that the canonical spinor $\psi_+=\id$ is not Majorana, and $B_1(R(\id),\id)= dx^2+dx^5$. Thus, when the spinor $\psi_+$ is not Majorana-Weyl there is a real direction in $M(\psi_+)$ that can be recovered by computing $B_1(R(\psi_+),\psi_+)$, as well as two complex directions that are the other two factors in $B_3(\psi_+,\psi_+)$. This gives one of the two types of orbits in $S_+$. 
 
 The other possible orbit is that of Majorana-Weyl spinors. For these spinors the only non-vanishing geometric object that can be constructed is $B_3(\psi_+,\psi_+)$, and it is given by the product of 3 real directions spanning $M(\psi_+)$ in this case. The group ${\rm Spin}(3,3)$ acts transitively on the space of Majorana-Weyl spinors, with the stabiliser isomorphic to ${\rm SL}(3)$ semi-direct product with $\R^3$. 
 
 For both Majorana-Weyl and general Weyl, a single spinor $\psi_+\in S_+$ defines only its MTN. To recover a complementary subspace, and thus a  structure of the mixed type, one needs another spinor $\psi_-$ such that $\langle\psi_+,\psi_-\rangle\not=0$. For example, we have
 \be
 B_2( - dz_{31} du, \id) = \im dx^1\wedge dx^4 + \im  dx^3\wedge dx^6 + dx^2 \wedge dx^5,
 \ee
 which is the mixed structure whose null eigenspaces are those on which the model is constructed. Similarly, for two Majorana-Weyl spinors we have
 \be
 B_2( -dz_{31} du - du, \id - dz_{31}) = 2 dx^1\wedge dx^6+ 2 dx^2\wedge dx^5+2 dx^3\wedge dx^4,
 \ee
 which gives a paracomplex structure with 3 real null directions $dx^1\pm dx^6, dx^2\pm dx^5, dx^4\pm dx^3$. 
  
 \section{Discussion}
 
 The main result of this paper is that there exists in general a number of creation/annihilation operator models of ${\rm Cl}(r,s)$. Different models correspond to different ways to split $\R^{r,s}=\R^{2k,2l}\oplus\R^{m,m}$. In addition, a model arises if a complex structure on the $\R^{2k,2l}$ factor, and a paracomplex structure on $\R^{m,m}$ is chosen. We have proposed to call the structure that encodes such a split, together with the complex/paracomplex structures on the factors, as a structure of the mixed type. The $\Gamma$-operators are then constructed as appropriate linear combinations of the creation/annihilation operators acting on polyforms built from the null $k+l+m$ eigendirections of the complex and paracomplex structures. Each model comes with two preferred pure spinors. These are the "empty" (function) and top polyforms. These two pure spinors have a non-vanishing inner product, and are complementary in the sense that their annihilator subspaces span $\R^{r,s}$. This pair of pure spinors allows one to recover the mixed structure that went into the construction of the model. This is obtained as $\langle \psi_1, \Gamma\Gamma\psi_2\rangle$, where $\psi_{1,2}$ are the two "canonical" pure spinors. In the opposite direction, each model corresponds to a pair $\psi_{1,2}$ of pure spinors of a fixed real index and satisfying $\langle \psi_1, \psi_2\rangle\not=0$. 
 
 This generalises the well-known story for ${\rm Cl}(2n)$ and ${\rm Cl}(n,n)$ to the case of a general ${\rm Cl}(r,s)$. The additional bonus is that our construction makes it clear that there are in general many different possible models even for the case of ${\rm Cl}(n,n)$. Indeed, the model based on a paracomplex structure on $\R^{n,n}$ is only one of the many possible models, the model with the maximal real index. Another canonical model that always exists is that of the minimal real index, which is either $0$ or $1$. 
 
Another point that is worth stressing is that, while for some choices $(r,s)$ it is possible to impose the Majorana-Weyl condition, the arising Majorana-Weyl spinors capture only some of the possible types of spinors of ${\rm Spin}(r,s)$. For example, for ${\rm Spin}(2,2)$ there are two types of orbits in the space of spinors. In one, the maximally-isotropic (or maximally-totally null (MTN)) subspace corresponding to a spinor is spanned by two real vectors. In the other, the MTN subspace is spanned by two complex vectors. Only one of this orbits is compatible with the Majorana-Weyl condition. Thus, the Majorana-Weyl condition, in situations when it becomes possible, restricts the geometry that corresponds to spinors by throwing away some cases that are possible and interesting. 

Probably the most interesting outcome of our analysis is the existences of structures mixing complex and paracomplex structures. As we have demonstrated, these mixed structures arise very naturally from (a pair of) pure spinors of a given real index. Such more general structures are only possible when the signature of the metric is not definite, and so it may seem that they are not of interest in Riemannian geometry. However, metrics of split signature $(n,n)$ do appear in the context of generalised geometry \cite{Hitchin:2003cxu}. Thus, a generalised complex structure on a manifold $M$ of dimension $2k$ is a complex pure spinor $\psi$ of ${\rm Spin}(2k,2k)$ (which is a complex polyform on $M$) that has the property $\langle R(\psi),\psi\rangle\not=0$ (and thus defines an almost complex structure on $TM\oplus T^*M$), and which has the integrability property that the null eigenspaces of the complex structure defined by $\psi$ are closed under the Courant bracket. The integrability condition can be shown to be equivalent to the condition that the polyform $\psi$ is closed on $M$, see \cite{Hitchin:2003cxu}. Given that there are several different types of pure spinors of ${\rm Spin}(2k,2k)$ (with different real index), with a pair of complementary pure spinors $\psi_{1,2}: \langle\psi_1, \psi_2\rangle\not=0$ defining in general a structure of a mixed type on $TM\oplus T^*M$, it would be interesting to study the arising more general types of geometric structures on $M$, thus further generalising the generalised geometry of \cite{Hitchin:2003cxu}.

\section*{Acknowledgements}

KK is grateful to F. Reese Harvey for correspondence.


\begin{thebibliography}{99}

\bibitem{BrauerWeyl} R. Brauer and H. Weyl, "Spinors in n dimensions", American Journal of Mathematics, Vol. 57, No. 2 (Apr., 1935), pp. 425-449.

\bibitem{Cartan} \'Elie Cartan, "The theory of spinors", The M.I.T. Press, Cambridge, Mass., 1967, translation of 1938 French original.

\bibitem{Chevalley} C. Chevalley, "The algebraic theory of spinors", Columbia Univ. Press, New York, 1954. 

\bibitem{BT} P. Budinich and A. Trautman, "Fock space description of simple spinors", SISSA, 1989. 

\bibitem{TT} A. Trautman and K. Trautman, "Generalized pure spinors", J. Geom. Phys. {\bf 15} (1994) 1-22. 

\bibitem{Harvey} F. Reese Harvey, "Spinors and Calibrations", Academic Press, 1990.

%\cite{DiLuzio:2011mda}
\bibitem{DiLuzio:2011mda}
L.~Di Luzio,
``Aspects of symmetry breaking in Grand Unified Theories,''
[arXiv:1110.3210 [hep-ph]].
%23 citations counted in INSPIRE as of 24 Mar 2021

\bibitem{BK2} N.~Bhoja and K.~Krasnov, "Notes on spinors and polyforms II: Quaternions and octonions."

\bibitem{KT} W. Kopczynski and A. Trautman, "Simple spinors and real structures", J. Math. Phys. {\bf 33} (2) (1992) 550-559. 

%\cite{Hitchin:2003cxu}
\bibitem{Hitchin:2003cxu}
N.~Hitchin,
``Generalized Calabi-Yau manifolds,''
Quart. J. Math. \textbf{54} (2003), 281-308
doi:10.1093/qjmath/54.3.281
[arXiv:math/0209099 [math.DG]].
%618 citations counted in INSPIRE as of 11 Mar 2022
 
\bibitem{Bryant}
R.~Bryant, "Pseudo-Riemannian metrics with parallel spinor fields and vanishing Ricci tensor," 
[arXiv:math/0004073 [math.DG]].
  


\end{thebibliography}
\end{document}